\documentclass[a4paper,10pt]{article}
\pdfoutput=1
\usepackage{fullpage}
\usepackage{cite}
\usepackage{amsfonts}
\usepackage{textcomp}
\usepackage{xcolor}
\usepackage{graphicx}
\usepackage{color}
\usepackage{amsmath}
\usepackage{bbm}
\usepackage{amssymb}
\usepackage{amsthm}
\usepackage{xspace}
\usepackage{hyperref}
\usepackage{pdfsync}
\usepackage{subcaption}
\usepackage{xspace}
\usepackage{algorithm}
\usepackage[noend]{algpseudocode}
\usepackage{titling}

\newcommand{\moveall}{{\sc Move-All}\xspace}

\newcommand{\visitall}{{\sc Visit-All}\xspace}

\newtheorem{lemma}{Lemma}
\newtheorem{theorem}{Theorem}

\thanksmarkseries{arabic}

\begin{document}

\title{Oblivious  Permutations on the Plane}

%
%

\author{Shantanu Das\thanks{Aix-Marseille University and CNRS, LIS, Marseille, France.  \emph{shantanu.das@lis-lab.fr}}, Giuseppe A. Di Luna\thanks{DIAG, University of Rome ``Sapienza'', Rome, Italy.  \emph{diluna@diag.uniroma1.it}}, Paola Flocchini\thanks{SEECS, University of Ottawa, Canada. \emph{flocchin@site.uottawa.ca}}, Nicola Santoro\thanks{School of Computer Science, Carleton University.   \emph{santoro@scs.carleton.ca}}, \\Giovanni Viglietta\thanks{JAIST, Nomi City, Japan. \emph{johnny@jaist.ac.jp}}, Masafumi Yamashita\thanks{Kyushu University, Fukuoka, Japan. \emph{masafumi.ymashita@gmail.com}}}
\date{}
\maketitle

\begin{abstract}
We consider a distributed system of $n$ identical mobile robots operating in the two dimensional Euclidian plane. As in the previous studies, we consider the robots to be anonymous, oblivious, dis-oriented, and without any communication capabilities, operating based on the Look-Compute-Move model where the next location of a robot depends only on its view of the current configuration. Even in this seemingly weak model, most formation problems which require constructing specific configurations, can be solved quite easily when the robots are fully synchronized with each other. In this paper we introduce and study a new class of problems which, unlike the formation problems so far,   cannot always be solved even in the fully synchronous model with atomic and rigid moves. This class of problems requires the robots to permute their locations in the plane. In particular, we are interested in implementing two special types of permutations -- permutations without any fixed points and permutations of order $n$. The former (called \moveall) requires each robot to visit at least two of the initial locations, while the latter (called \visitall) requires every robot to visit each of the initial locations in a periodic manner. We provide a characterization of the solvability of these problems, showing the main challenges in solving this class of problems for mobile robots. We also provide algorithms for the feasible cases, in particular distinguishing between one-step algorithms (where each configuration must be a permutation of the original configuration) and multi-step algorithms (which allow intermediate configurations). These results open a new research direction in mobile distributed robotics which has not been investigated before. 
\end{abstract}

\section{Introduction}

The investigation of the computational and complexity issues arising in distributed systems of autonomous mobile robots is an important research topic 
in distributed computing.
This has several applications, teams of robots could  be sent to regions inaccessible to humans to perform a variety of tasks 
such as exploration and data-collection, monitoring, sensing or patrolling. Once deployed, the team of robots must coordinate 
with each other and perform the tasks autonomously without human intervention; this has motivated the design of distributed algorithms for coordination among the robots to enable them to perform the required tasks.

As a theoretical abstraction, the robots are usually viewed as computational entities modelled as points in a metric space, typically $\mathbb
R^2$, in which they can move. The robots, identical and outwardly indistinguishable, have the same capabilities and execute the same (deterministic) algorithm.
They can see each other, but cannot explicitly communicate with one another. This lack of direct communication capabilities means that the only means of interaction between robots are observations and movements: that is, communication is {stigmergic}. Each robot operates in ``Look-Compute-Move'' (LCM) cycles: during a cycle, it observes its surroundings, computes a destination point, and moves to it. Typically, the robots are assumed to  have constant-size persistent memory or, more commonly, to be {\em oblivious} having no persistent memory: This paper assumes the latter model where robots in each cycle act only based on the current observation and have no memory of their activities from previous cycles. Further the robots do not have any means of orienting themselves; 
Each robot observes the location of other robots relative to its own position in the plane and the robots do not share any common coordinate system.
If the robots agree on a common notion of clockwise direction, then we say the system has \emph{chirality}.

Some typical problems that have been studied in this model include: {\em gathering} of robots (e.g., \cite{DiFlSV18,DiFlSVY17}), uniform {\em dispersal}, {\em filling} a region with robots, {\em flocking}, etc.
(for a review, see \cite{FloPS12}). A generalization of some of these problems is that of \emph{pattern formation}, where the $n$ robots need to move 
from any initial configuration to a predefined pattern of $n$ points in the plane. This class has been extensively studied 
(e.g., \cite{andoSY95,BrT16,DaFSY15,FPSW08,FuYaKiYa12,SugS96,SuzY99,yamashita2010,YuY14}).  
A major issue in such formation problems is the amount of symmetry (quantified by the notion of symmetricity~\cite{SuzY99}) in the starting configuration of robots and in the points of the pattern.
In the arbitrary pattern formation problem, the points where the pattern is formed are  {\em relative}, i.e. subject to rotation, translation and scaling of the input pattern. 
A different line of research is when the points of the pattern are {\em fixed}, a setting called {\em embedded pattern} and studied in \cite{CiDN16,FuYOKY15}. 
 
In some applications, forming a pattern may be the first step of a more complex task requiring coordination between robots. Consider, for example, robots that contain instrumentation for monitoring a site once there, as well as sensors for measurement (e.g., detecting traces of oil or precious metals, radioactivity, etc).
If each robot has different sensors,  the same site  might need to be visited by all robots, and this must be done
while  still keeping all the sites monitored. 
A more relaxed version of this task is where  each site must be visited by  (at least) two robots. This task may be  useful even in situations where  all the robots contain the same sensors,  e.g., if there are faulty sensors and we want to replicate the measurements. 

These  tasks are instances of a new class of problems quite different from the formation problems as the robots need to rotate among the given points of interests, forming permutations of a given pattern of points. We assume that each robot is initially occupying a point of interest (thus marking that location) and the objective is to permute the robots among these locations periodically. The question is which permutations can be implemented starting from which patterns. We show a big difference between these classes of permutation problems compared to the formation problems studied previously. In particular, we show that even in the {\em fully synchronous} (${\cal FSYNC}$) model, some of the permutation problems are not solvable, even when starting from configurations that admit a leader. In contrast, any formation problem (including gathering) is easily solvable in ${\cal FSYNC}$ when the starting configuration admits a leader.  

Note that the permutation problems considered in this paper are {\em perpetual} tasks requiring continuous visits to the sites by the robots.
Unlike the multiple pattern formation problem where robots continuously move from a pattern to the next \cite{DaFSY15},  here the robots perpetually move but only  exchanging locations in the same pattern. 
In particular, we focus on two interesting types of permutations --- permutations without fixed points, and permutations of order $n$ (i.e. $n$-cycles). These give rise to two specific problems (i) \moveall :  every site must be visited by at least two robots and every robot has to visit at least two points, and,
(ii) \visitall : every robot must visit each of the points of interest. 
We provide a characterization of the solvability of these problems showing which patterns make it feasible to solve these problems and under what conditions. 
To the best of our knowledge, this is the first investigation on these class of problems.

\subsection*{Our Contributions}
We distinguish between $1$-step and multi-step algorithms; In the former case, we must form the permutations without passing through intermediate configurations, while in the latter case, a fixed number of intermediate configurations are allowed (see definitions in Section~\ref{sec:model}). We study $1$-step and 
$2$-step algorithms for \visitall and \moveall, distinguishing the case when the robots share a common chirality from the case when they do not.
We identify a special class of configurations denoted by ${\cal C}_{\odot}$, that are rotationally symmetric with exactly one robot in the center of symmetry. Such  configurations do not always allow permutations without fixed points, thus making it difficult to solve the above problems. 

We show that when there is chirality, the sets of initial configurations from which \visitall and \moveall can be solved, using $1$-step algorithms, are the same: that is, all configurations except those in  ${\cal C}_{\odot}$ (Section \ref{obv:chiral}).
We then show that the characterization remains the same when we consider $2$-step algorithms. Moreover, in the case of \visitall, the solvability does not change even for $k$-step algorithms for any constant $k$. 

On the other hand, when there is no chirality, we observe a difference between the solvability of  \visitall and \moveall. Configurations in ${\cal C}_{\odot}$ are clearly still non feasible for both problems.  However, for the \moveall problem  the class of unsolvable configurations also includes the ones  where there exists a symmetry axis with a unique robot on it. 
On the other hand, the set of initial configurations from which \visitall is solvable is different:  the problem can be solved  if and only if in the initial configuration there are no axes of symmetry or if there is a unique symmetry axis that does not contain any robots. 
Interestingly, also in this case, allowing $2$-step algorithms does not change the set of solvable instances. 

We then show  that, when there is chirality and the coordinate systems of robots are visible (that is, a robot can sense the local coordinate system of the others), then \visitall (and thus  \moveall) is solvable from arbitrary initial configurations, and we provide a universal algorithm for solving the problems. 
Finally, we show that allowing a single bit of persistent memory per robot and assuming chirality, it is possible to solve the problems for all initial configurations (Section \ref{sec:mem}).

\section{Model, Definitions and Preliminaries}\label{sec:model}

\noindent {\bf Robots and scheduler.}
We consider a set of dimensionless computational entities: the {\em robots}. These robots are modelled as points in the metric space $\mathbb{R}^2$; they are able to sense the environment detecting the presence of other robots, they can perform computations, and are able to move to any other point in the space. Each robot has its own local coordinate system centred in its own position (which may differ in orientation and unit distance from the coordinate system of other robots). For simplicity of description, we will use a global coordinate system $S$ for analyzing the moves of the robots (robots themselves are unaware of this global system).
Robots are {\em oblivious}: they do not have any persistent memory and thus, they cannot recall any information from previous computations. We indicate the set of robots with $R:\{r_0,r_1,\ldots, r_{n-1}\}$, however the robots themselves are not aware of the numbering assigned to them. All robots are identical and follow the same algorithm. 
We assume the so-called {\em Fully-Synchronous Scheduler} (${\cal FSYNC}$). Under this scheduler, time can be seen as divided in discrete fixed length slots called  {\em rounds}. In each round, each robot synchronously performs a {\em Look-Compute-Move} cycle \cite{FloPS12}. During the {\em Look} phase, a robot $r$ takes an instantaneous {\em snapshot} of the environment, the snapshot is an entire map of the plane containing positions of all the other robots with respect to the local coordinate system of $r$. During the {\em Compute} phase, robot $r$ performs some local computation to decide its new destination point as a function of the aforementioned snapshot as input. Finally, in the {\em Move} phase, the robot moves to the computed destination point (which may be the same as current location).

\noindent {\bf Chirality.}
Robots may or may not share the same {\em handedness}: in the former case, they all agree on the clockwise direction and we say the system has {\em chirality} \cite{FloPS12}, in the latter case, robots do not have such an agreement and we say there is {\em no chirality}. 

\noindent {\bf Configurations.}
A configuration $C$ is an ordered tuple of points $C = (p_0, p_1, \ldots , p_{n-1})$,
where $p_i=C[i]$ is the position of robot $r_i$ in terms of the global coordinate system $S$. We denote by $Z = (Z_0, Z_1, \ldots Z_{n-1})$ the ordered tuple of coordinate systems where $Z_i$ is the system used by robot $r_i$. 
Given a robot $r_i$ located at $p_i$, we denote with $C \setminus \{r_i\}$ (or sometimes $C \setminus \{p_i\}$), the configuration obtained by removing robot $r_i$ from $C$.
We indicate with $C_0$ the initial configuration in which the robots start. We denote by $SEC(C)$ the smallest circle that encloses all points in the configuration $C$.

\noindent {\bf Symmetricity.} 
Given any configuration $C$ with robots having coordinate systems $Z$, the {\em symmetricity} $\sigma(C, Z) = m$ is the largest integer $m$ such that the robots can be partitioned into classes of size at most $m$ where robots in the same class have the same view (snapshot) in $C$ (See \cite{SuzY99,yamashita2010}). Alternatively, we can define the symmetricity (irrespective of $Z$) of a configuration as $\rho(C)=m$ where $m$ is largest integer such that $\exists{Z} : \sigma(C, Z)=m$.
For any configuration $C$, we have $\rho(C) \geq 1$, the configurations with $\rho(C)=1$ are considered to be asymmetric (these are the only configurations that allow to elect a leader among robots).
For symmetric configurations with $\rho(C)>1$, $C$ may have rotational symmetry with respect to the center $c$ of $SEC(C)$, which coincides with the centroid of $C$ in this case, or $C$ may have mirror symmetry with respect to a line, called the axis of symmetry.
See Figure \ref{cdot} for an example of symmetry classes.

We define a special class of configurations 
denoted by ${\cal C}_{\odot}$.
A configuration $C$ is in ${\cal C}_{\odot}$, if and only if $\rho(C)=1$, and there exists a unique robot $r_c$ (the {\em central robot}) located at the center of $SEC(C)$ such that $\rho(C \setminus  \{r_c\})=k > 1$; In other words, 
$C$ has a rotational symmetry around $r_c$ such that $C$ can be rotated around centre $r_c$ by an angle $\theta =\frac{\pi}{k}$ to obtain a permutation of $C$. Figure \ref{imptriang} is an example of a configuration in ${\cal C}_{\odot}$). 

\begin{figure}
  \centering
    \includegraphics[width=0.4\textwidth]{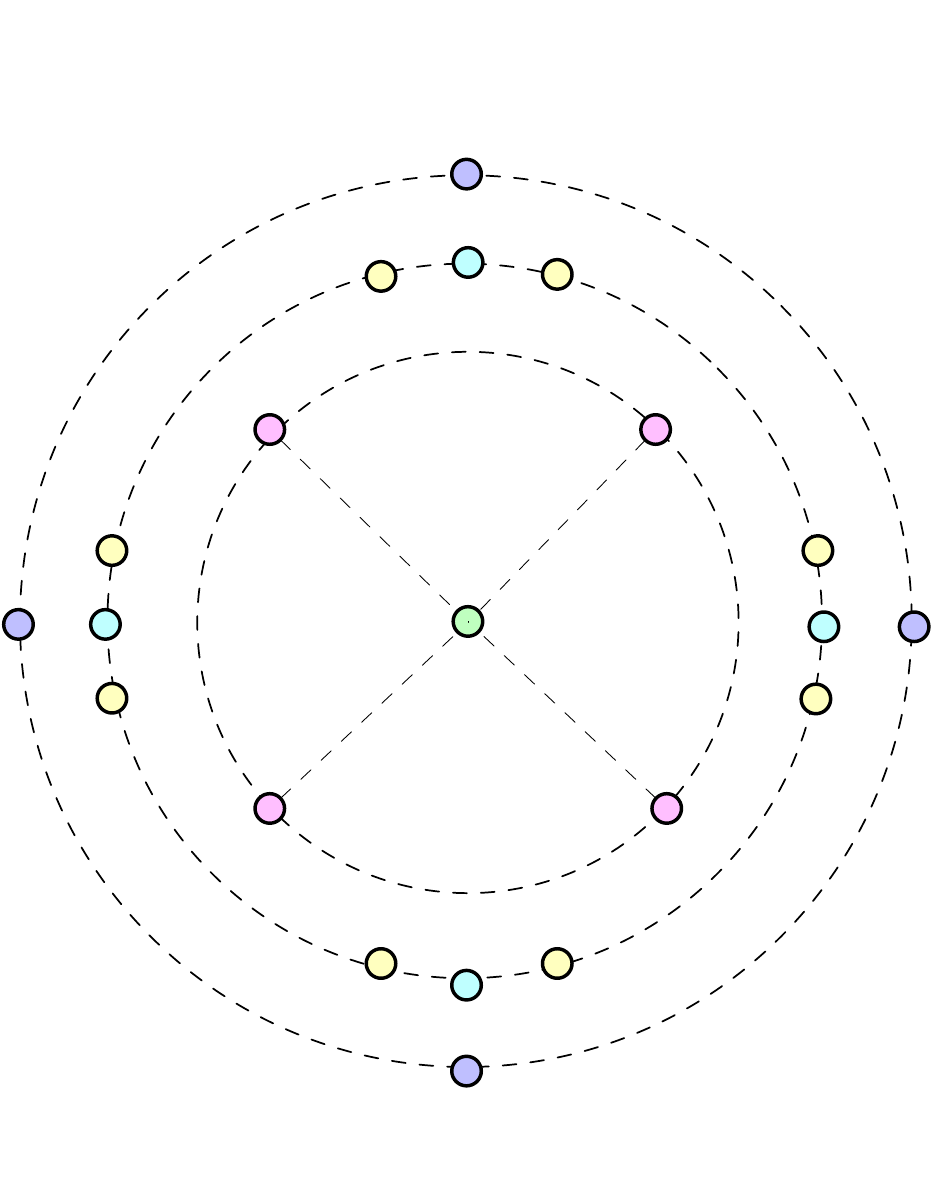}
      \caption{Configuration $C \in {\cal C}_{\odot}$, each symmetry class is coloured in a different way. Notice that the central robot $r_c$ is the only one with an unique view (recall that the central robot is the green one, the robot positioned on the centroid of $C$). If removed we get a configuration $C \setminus \{r_c\}$, with $\rho(C \setminus \{r_c\})=4$, notice also that the yellow class has cardinality $8$.}
      \label{cdot}
\end{figure}

\noindent {\bf Permutations and runs.} For a permutation $\pi = (\pi(0), \pi(1), \ldots , \pi(n-1))$ of $(0,1, \ldots , n-1)$,
define $\pi(C) = (p_{\pi(0)}, p_{\pi(1)}, \ldots , p_{\pi(n-1)})$.
We denote: (1) the set of permutations with no fixed points as $\Pi_2 = \{ \pi : \pi(i) \neq i ,\,\,\forall i: 0 \leq i \leq n-1 \}$ and (2) the set of cyclic permutations of order $n$ as
$\Pi_n = \{ \pi : \pi^j(i) = i <=> nk = j$ for some $k \in \mathbb{N} \}$ where $\pi^{j}$ indicates that we apply permutation $\pi$  $j$ times. 
Let $\Pi(C)$ be the set of all permutations of ${C}$.

Given an algorithm ${\cal A}$  and an initial configuration $C_0$ we denote any execution of algorithm {\cal A}, starting with configuration $C_0$ as the run ${\cal R}_{{\cal A},C_0}:(C_0,C_1,C_2,\ldots)$, the infinite ordered sequence of configurations, where $C_j$ is produced during round $j$ of the execution. 


\subsection*{Problem Definitions} 

We will study the following two problems:
\begin{itemize}
\item \moveall: An algorithm ${\cal A}$ is a $1$-step solution algorithm for the \moveall problem, if every possible run of the algorithm ${\cal R}_{{\cal A},C_0}:(C_0,C_1,C_2,\ldots)$ is such that: $C_{i} = \pi^{i}(C_0)$ for some $\pi \in \Pi_2$. Intuitively, every configuration is a permutation of $C_0$ and in any two consecutive configurations, the position of each robot is different. 
As an extension for any $k \in \mathbb{N}^{+}$, a $k$-step solution requires that $C_{i \cdot k} = \pi^{i}(C_0)$ where $\pi \in \Pi_2$. (There is no constraint on the intermediate configurations $C_j$ where $k$ does not divide $j$.)

\item \visitall: 
An algorithm ${\cal A}$ is a $1$-step solution algorithm for the \visitall problem, if every possible run of the algorithm ${\cal R}_{{\cal A},C_0}:(C_0,C_1,C_2,\ldots)$ is such that: $C_{i} = \pi^{i}(C_0)$ for some $\pi \in \Pi_n$. Intuitively, every configuration is a permutation of $C_0$ and in every $n$ consecutive configurations, every robot visit every location $p_i \in C_0$.
We can similarly define a $k$-step solution for the problem where $C_{i \cdot k} = \pi^{i}(C_0)$ for some $\pi \in \Pi_n$.


\end{itemize}
Since $\Pi_n \subset \Pi_2$, it follows that any solution for \visitall is also a solution to the \moveall problem.

\subsection*{Oblivious Permutations}

Note that $k$-step solutions of \moveall and \visitall specify that we must have a permutation of the initial configuration $C_0$ every $k$ rounds. However, no constraint is given on the other {\em intermediate} configurations. Interestingly, when robots are oblivious the previous definitions imply a stronger version of the problem in which each configuration $C_{j+k}$ has to be the permutation of configuration $C_j$ that appeared $k$ rounds ago. 

\begin{lemma}\label{lemma:prem}
Let ${\cal A}$ be a $k$-step algorithm solving \moveall (or \visitall), and let ${\cal R}_{{\cal A},C_0}:(C_0,C_1,C_2,\ldots)$ be any run of ${\cal A}$ starting from $C_0$.  For each $j \in \mathbb{N}$ we have that $C_{j+k} =\pi(C_j)$ for some $\pi \in \Pi(C_j)$. 
\end{lemma} 

\begin{proof}
We prove the lemma for \moveall, the extension to \visitall is analogous and immediate.
If $j=t\cdot k$ for some $t \in \mathbb{N}$ then the lemma follows from the problem definition. Thus let us consider a configuration $C_{j}$ such that $j \neq t\cdot k$ for all $t \in \mathbb{N}$. 
We observe that ${\cal R}_{{\cal A},C_j}$ (that is a run of ${\cal A}$ starting from $C_j$) is equal to the suffix of ${\cal R}_{{\cal A},C_0}$ starting from $C_j$. This is due to the obliviousness of the robot, the fact that the algorithm is deterministic and the synchronous scheduler: starting from a certain configuration and an assignment of local coordinate systems, the algorithm will generate a fixed sequence of configurations. However in  ${\cal R}_{{\cal A},C_j}$ we must have that $C_{j+k} = \pi(C_j)$ for some $\pi \in \Pi_{2}(C_j)$, otherwise ${\cal A}$ is not a correct algorithm for \moveall. 
\end{proof}

\section{Oblivious Robots with Chirality \label{obv:chiral}}

In this section we  consider robots having chirality (i.e., they agree on the same clockwise orientation). 

\subsection{$1$-Step Algorithms \label{1stepchirality}}

We first consider $1$-step algorithms, and show that \moveall and \visitall are solvable if the initial configuration $C_0$ is not in ${\cal C}_{\odot}$. 

\noindent {\bf Intuition behind the solution algorithms.}
The underlining idea of our solution algorithms is to first make robots agree on a cyclic ordering of the robots, and then permute their positions according to this ordering. This algorithm is shown in Algorithm \ref{visitciclyc} and the ordering procedure is shown in Algorithm \ref{alg:visitallrho2}. When the centroid $c$ of configuration $C_0$ does not contain any robot, we compute a cyclic ordering on the robots by taking the half-line passing through $c$ and one of the robots closest to $c$ and rotating it w.r.t. point $c$; the robots are listed in the order the line hits them. We can show that the ordering computed by any robot is a rotation of that computed by another robot (See Figure~\ref{figure:alg1} for example). The only issue is when there is a robot positioned in the centroid. In this case, the robots compute a unique total order on the robots; this is always possible since $C_0 \notin {\cal C}_{\odot}$, which implies that $C_0$ is asymmetric and admits a total ordering.

From the above observations, It is immediate that the algorithm solves \visitall, take a robot $r$, w.l.o.g. in position $p_i$, during $n$ activations, the robot moves through all the robot positions in the computed cyclic order, returning back to $p_i$; thus, it has visited every point in $C_0$.

\begin{algorithm}[h]
 \caption{{\sc \visitall}  Algorithm using a cyclic order. \label{visitciclyc}}
 \footnotesize
\begin{algorithmic}[1]
\State Compute a cyclic order $(p_0,p_1,\ldots,p_{n-1})$ on $C$ using {\sc Order}$(C)$.    
\State If my position is   $p_i$,  set {\bf destination}  as   $p_{(i+1) \mod {n}}$.   \label{order}
\end{algorithmic}
\end{algorithm}

\begin{algorithm}[h]
\caption{{\sc Order} algorithm with Chirality. \label{alg:visitallrho2}}
\footnotesize

\begin{algorithmic}[1]

\Procedure{Insert}{list, Configuration $C$} \Comment{ if this procedure is called then $\rho(C_0) = 1$.}
\State Let $order$ be a total order of robots in $C$ with the centroid $c$ as last element, and let $p$ the position that precedes $c$ in $order$.
\State insert $c$ in list after position $p$
\State {\bf return} list
\EndProcedure
\bigskip

\Procedure{Polar}{position $p$, reference $r$, Configuration $C$}
\State  {\sc Polar} takes two points $p,r$, and a configuration $C$.
\State  {\sc Polar} returns the polar coordinates $(d,\theta)$ of $p$ in a coordinate system that is centered in the centroid of $C$ and has as reference direction the segment between the centroid and point $r$.
\State {\bf return} $(d,\theta)$ 
\EndProcedure
\bigskip

\Procedure{Next}{position $r$, Configuration $C$}
\If{$\exists p \in C$ {such that $(d',0)=${\sc Polar}($p$,$r$,$C$) with $d' > ||r||$}  }
\State {\bf return} $p$
\Else
\State Let $p \in C$ be such that $(d,\theta)=${\sc Polar}($p$,$r$,$C$) has the minimum $\theta > 0$ and the minimum $d$ among all the polar coordinates of robots in $C$.
\State {\bf return} $p$
\EndIf
\EndProcedure
\bigskip

\Procedure{Order}{Configuration $C$}
\State $c$=centroid of $C$
\State list=EmptyList()
\State $r_1$=pick one robot position in $C$ with minimum non zero distance from $c$.
\State list.Append($r_1$)
\State $A=C \setminus \{r_1,c\}$
\While{$A \neq \emptyset $}
\State $r_2$={\sc Next}(list.LastElement(),$C$)
\State list.Append($r_2$)
\State $A=A \setminus \{r_2\}$
\EndWhile
\If{$c \in C$}
\State list={\sc InsertCentroid}(list,$C$) \label{ref:insertc}
\EndIf
\State return list
\EndProcedure
\end{algorithmic}
\end{algorithm}

\begin{figure}[h]
\begin{center}

  \begin{subfigure}[b]{0.45\linewidth}

  \fbox{\includegraphics[width=\textwidth]{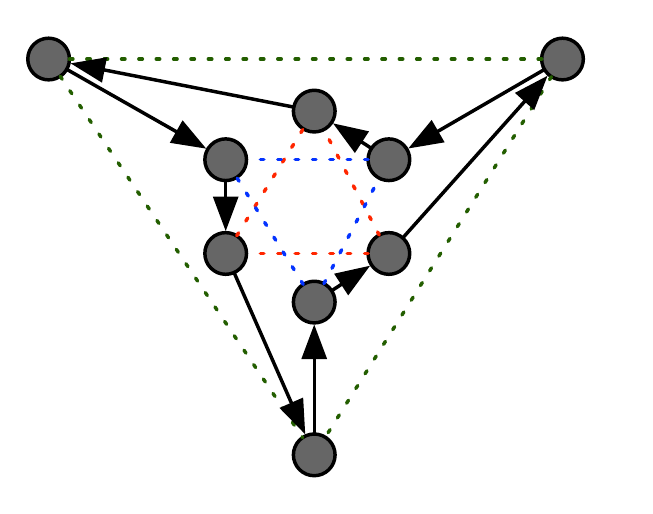}}
    \caption{An example of cyclic order induced by the {\sc Order} algorithm.    \label{fig:alg1c3}}
  \end{subfigure}
  \quad\quad
    \begin{subfigure}[b]{0.36\linewidth}

  \fbox{\includegraphics[width=\textwidth]{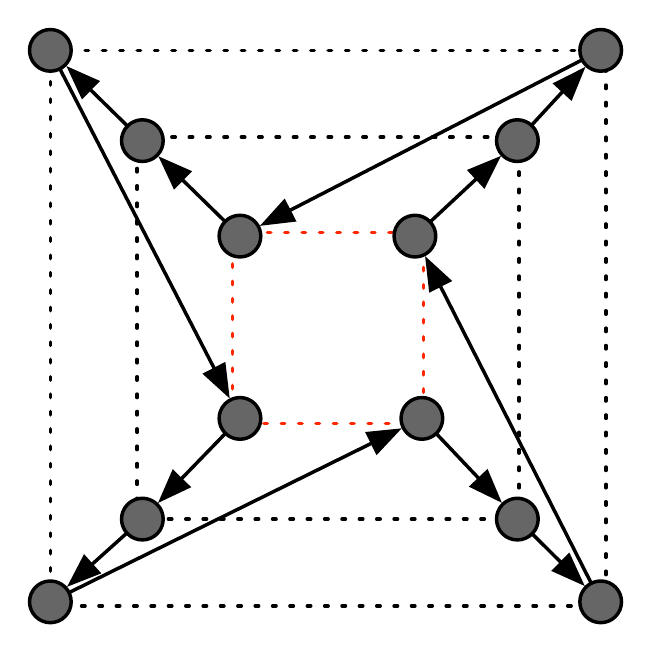}}
    \caption{Another example of cyclic order induced by the {\sc Order} algorithm.}
    \end{subfigure}
\end{center}
\caption{{\sc Order} algorithm: Examples of cyclic order computed by Algorithm \ref{visitciclyc}. \label{figure:alg1}}

\end{figure}

\begin{theorem}\label{lemma:l1}
In systems with chirality, there exists a $1$-step algorithm that solves \visitall from any initial configuration $C_0 \not\in {\cal C}_{\odot}$.
\end{theorem}

\begin{proof}
The proof is constructive; in fact, we show a $1$-step algorithm  that solves \visitall.
Each robot obtains a cyclic order by executing Algorithm \ref{alg:visitallrho2}. To prove that \visitall can be solved, it is sufficient to show that the cyclic orders are the same for all robots. 

The order is decided by Procedure {\sc Order}. We distinguish two cases, when the centroid $c$ of $C_0$ is not in $C_0$, and when it is in $C_0$.
\begin{itemize}
\item Case of $c \not\in C_0$:
Let us consider by contradiction that for two different robots, w.l.o.g. $r_1,r_2$, there are two different cyclic orderings. Let $L_1$ be the list given by {\sc Order} of robot $r_1$ and let $L_2$ be the list given by {\sc Order} of robot $r_2$.
Being $L_1$ and $L_2$ different, then there exist two indices $j_1,j_2$ such that $L_1[ (j_1)\mod n]=L_2[(j_2) \mod n]$ and $L_1[ (j_1+1)\mod n] \neq L_2[(j_2+1) \mod n]$.
However, this implies that $\exists p_j \in C_0$, such that {\sc Next}($p_j,C_0$) gives two different results according to the coordinate system of the callee, that is impossible: the centroid is the same for each one of them, moreover robots agree on the same handedness therefore the angles grows  counterclockwise for all of them. 
 \item Case of $c \in C_0$: First of all, note that the argument of the previous case implies that the $list$ obtained before executing Line \ref{ref:insertc} of Algorithm \ref{alg:visitallrho2} represents a common cyclic order on $C_0 \setminus \{c\}$  shared by all robots.  It remains to show that, when executing Line \ref{ref:insertc}, all robots insert the centroid $c$ before the same position $p$ in $list$. We have by assumption that $C_0 \not\in {\cal C}_{\odot}$, thus if $c \in C_0$ then there exists a total order on which all robots agree: by definition we either have $\rho(C_0 \setminus \{c\})=1$ (and the total order is trivially enforced) or that there is an unique symmetry axis of $C$; in the latter case the presence of chirality ensures such total order. It is easy to see that from this total order we can obtain a total order on which $c$ is the last element, and on which all robots agree, let it be $order$.  Therefore, being $order$ common to all robots, it is immediate that everyone inserts $c$ before the same position $p$ when executing Line \ref{ref:insertc}. 
  \end{itemize}\end{proof} 

\noindent We now show that, when $C_0 \in {\cal C}_{\odot}$ \moveall (and thus \visitall ) is unsolvable:

\begin{theorem}\label{tm3}
If $C_0 \in {\cal C}_{\odot}$ there exists  no $1$-step algorithm that  solves \moveall, even when the robots have chirality.
\end{theorem} 

\begin{proof} In any configuration in ${\cal C}_{\odot}$, the adversary can assign coordinate systems in such a way that each robot, except the central robot $r_c$,  has at least one analogous with a symmetric view. This derives directly from the definition of ${\cal C}_{\odot}$. It is immediate to see that it is impossible to elect a  unique robot to move to the center  of $C_0$, taking the position of $r_c$. An example is given in Figure \ref{imptriang}, where if one robot moves to the centroid of $C_0$, then every robot except $r_c$ would do the same. This implies that, in the next round, it is impossible to form any $C_1 \in \Pi(C_0)$ with a different central robot.\end{proof}

\begin{figure}
  \centering
    \includegraphics[width=0.5\textwidth]{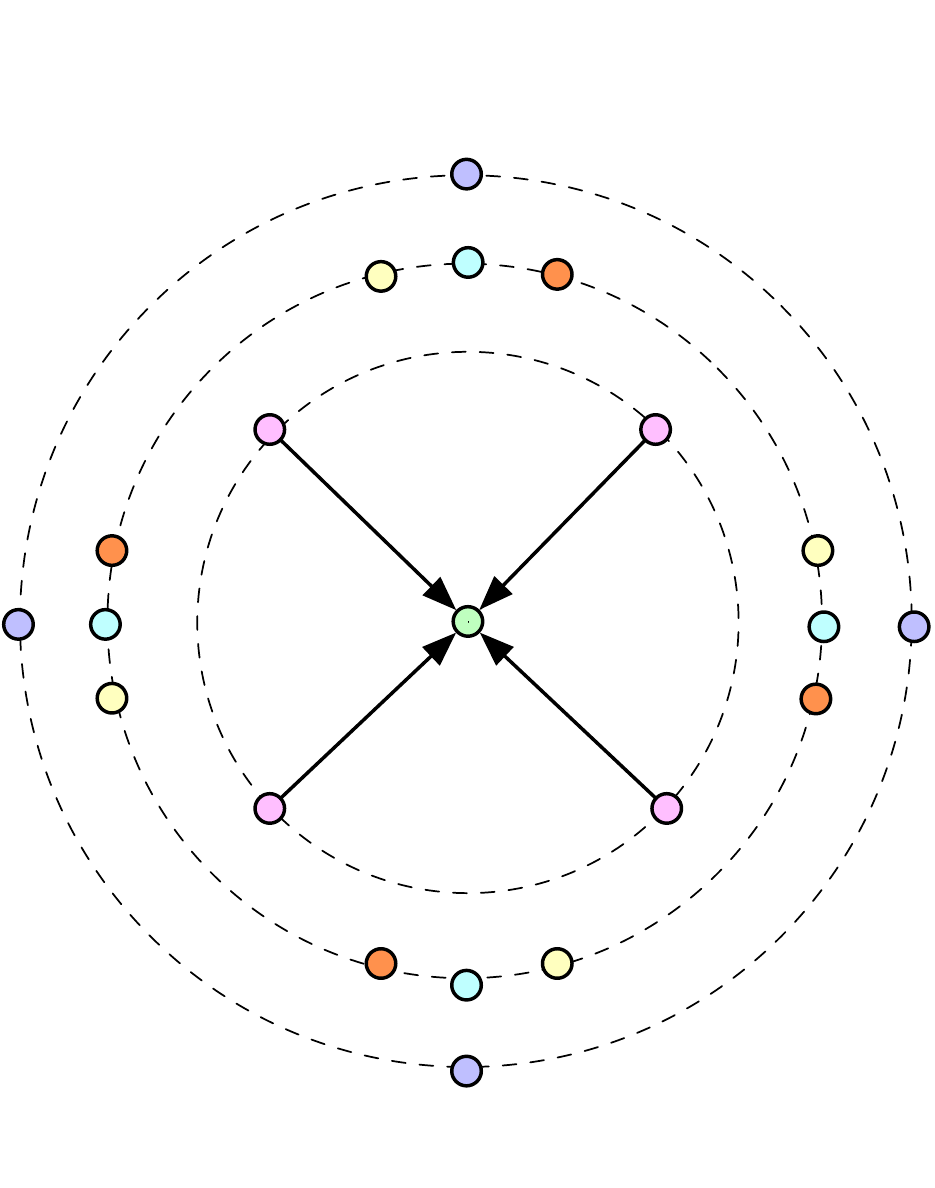}
      \caption{Configuration $C_0 \in {\cal C}_{\odot}$ where it is impossible to solve \moveall with a $1$-step algorithm.}
      \label{imptriang}
\end{figure}
Note that  Theorem  \ref{lemma:l1} implies that \moveall is solvable under the same assumptions of the theorem (recall that if we satisfy the \visitall specification, we satisfy also \moveall specification). Moreover, for the same reason, Theorem \ref{tm3} implies that \visitall is unsolvable. 

We can summarize the results of this section as follows:

\begin{theorem}\label{obs1}
In systems with chirality, \moveall and \visitall can be solved in 1-step if and only if  $C_0 \not\in {\cal C}_{\odot}$.
\end{theorem}

\subsection{$2$-step Algorithms}\label{chir:twostep}

In light of Theorem \ref{obs1}, one may wonder what happens when multiple steps are allowed. In this section we show that allowing an intermediate step to reach the goal does not bring any advantages. We first introduce a technical lemma. 
Intuitively, the result is based on the observation that it is impossible to replace the central robot by another robot in $1$-step. Thus the intermediate configuration must be a configuration $C_1 \notin {\cal C}_{\odot}$. 

\begin{lemma}\label{lm:sequence}
Let ${\cal A}$ be a $2$-step algorithm that solves \moveall.  Starting from configuration $C _0\in {\cal C}_{\odot}$, algorithm ${\cal A}$ cannot generate a run  ${\cal R}_{{\cal A},C_0}:(C_0,C_1, C_2,C_3,\ldots)$ where $C_1 \in {\cal C}_{\odot}$.
\end{lemma}

\begin{proof}
Being ${\cal A}$ a $2$-step algorithm we have $C_2=\pi(C_0)$ for some $ \pi \in \Pi(C_0)$, and being a solution for \moveall we have that the central robot, let it be $r_c$, in $C_0$ has to be different from the central robot in $C_2$. The proof is by contradiction. Let us assume $C_1 \in {\cal C}_{\odot}$, we will show that $r_c$ is also the central robot in $C_1$. 
Let us suppose the contrary, and let $r_x \neq r_c$ be the central robot in $C_1$. Note that, by definition of ${\cal C}_{\odot}$, robot $r_x$ is the only robot that has a view different from all the others in configuration $C_1$. This is equivalent to say that, in a configuration obtained by removing $r_c$ from $C_0$, robot $r_x$ may move to a position in such a way to make  its view unique and breaking the symmetry of the configuration. This directly  contradicts the fact that $C_0 \in {\cal C}_{\odot}$, in $C_0$ the only robot with a  unique view is $r_c$ and it is the only robot that can break the symmetry by moving. 
Therefore, the central robot in $C_1$ has to be $r_c$.
By using the same argument, we have that also the central robot of $C_2$ has to be $r_c$, which contradicts the correctness of ${\cal A}$.
\end{proof}

Based on the above result, we can show that is impossible to solve  \moveall from a configuration  $C_0 \in {\cal C}_{\odot}$,  even if the system has chirality.
The informal idea here is that the central robot $r_c$ in configuration $C_0$ needs to move away from the center to form the intermediate configuration $C_1$. However, in any $2$-step algorithm, $C_2$ must be a permutation of $C_0$, with a different robot $r'$ in the center. Now, following the same algorithm, robot $r'$ would move away from the center to form the next configuration $C_3$. By choosing the coordinate systems of robots $r_c$ and $r'$ in an appropriate way, the adversary can ensure that $C_3$ would not be a permutation of $C_1$, thus violating the conditions in Lemma~\ref{lemma:prem}.  The above reasoning is formalised in the following theorem:

\begin{theorem}\label{tmimp1} 
There  exist no $2$-step algorithm that solves \moveall from a configuration  $C_0 \in {\cal C}_{\odot}$,  even if the system has chirality.
\end{theorem} 

\begin{proof}
Consider a configuration $C_0 \in {\cal C}_{\odot}$ where all robots but the central are vertices of a regular polygon, as a reference see Figure \ref{fig:ncmove}. Let us assume that there exists a    $2$-step algorithm ${\cal A}$  that solves \moveall starting from $C_0$.  
Let $r_{c}$ be the central robot in $C_0$, with coordinate system $Z_{r_c}$, and let $NC:\{r_1,\ldots,r_{n-1}\}$ be the set of the other robots. There are two possible behaviours of ${\cal A}$ according to which a robot moves when ${\cal A}$ starts from configuration $C_0$:
\begin{itemize}

\begin{figure}[H]
  \centering
    \includegraphics[width=0.4\textwidth]{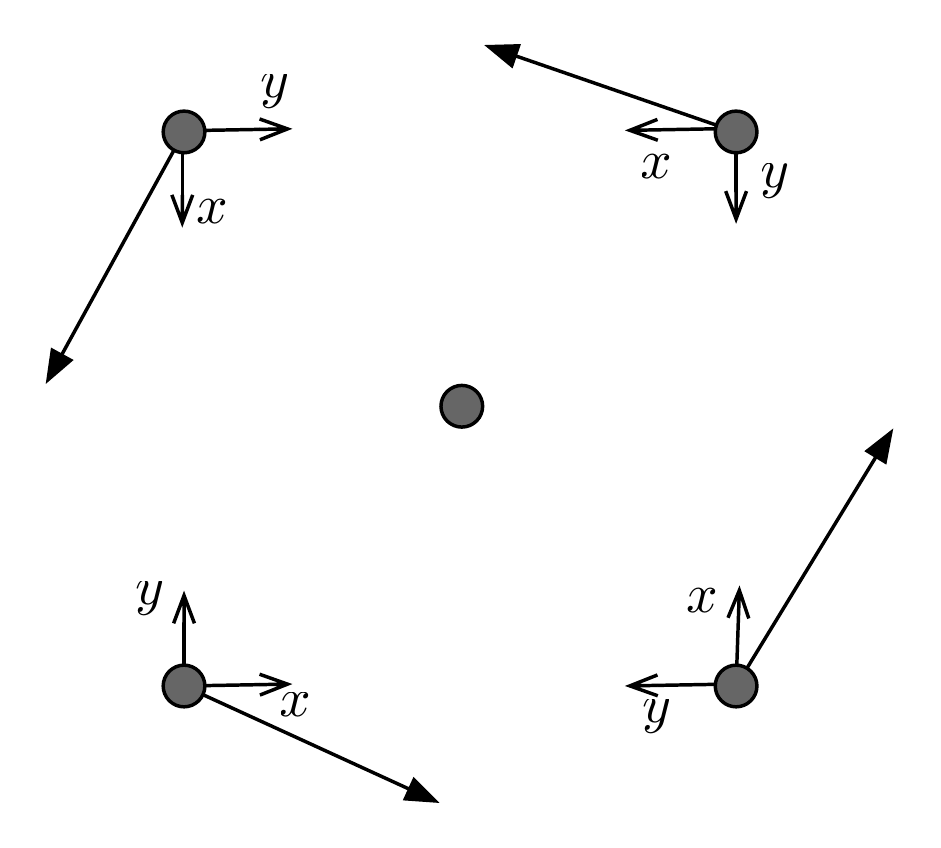}
      \caption{Case in which algorithm {\cal A} makes only (and all) robots in $NC$ move: the resulting configuration will be still in ${\cal C}_{\odot}$.} 
      \label{fig:ncmove}
\end{figure}

\item Robot $r_c$ moves, and possibly also robots in $NC$. Being ${\cal A}$ a $2$-step algorithm,
 starting from configuration $C_0 \in {\cal C}_{\odot}$, algorithm ${\cal A}$ generates a run $(C_0,C_1,C_2,C_3,\ldots)$ where $C_x \in \Pi(C_0)$ when $x$ is even, and $C_x \not\in {\cal C}_{\odot}$  when $x$ is odd (see Lemma \ref{lm:sequence}).  
Let us now assume that the coordinate system of any $r \in NC$ is obtained by rotating  the coordinate system of robot $r_c$ by 
$\frac{\pi}{2}$. 
Since  ${\cal A}$  must solve \moveall, we have that in $C_2$ the robot in the center of symmetry cannot be $r_c$, let it be $r$. Note that $r$ will move to a destination point that is different from the one $r_c$ moved starting from $C_0$, which means that  configuration $C_3 \not\in \Pi(C_1)$.
However, this leads to a contradiction, see Lemma \ref{lemma:prem}.

\item Only robots in $NC:\{ r_1,\ldots,r_{n-1}\}$ move.  First note that  all robots in $NC$ have symmetric views. Therefore, it is not possible to make only one of them move, everyone has to move. With this is mind is easy to see that, for any possible movement option of robots in $NC$, there exists an initial arrangement of local coordinate systems 
such that the resulting configuration will be a rotation and a scaling of the initial configuration, see Figure \ref{fig:ncmove} for an example. Therefore, configuration $C_1$ is still in ${\cal C}_{\odot}$, by Lemma \ref{lm:sequence} this implies that ${\cal A}$ is not correct.
\end{itemize}
Since ${\cal A}$ is not correct in any of the previous cases, we have that the existence of ${\cal A}$ is impossible. 
\color{black}
\end{proof}

Interestingly, when we consider \visitall we can prove a stronger impossibility result that includes algorithms that could depend on the number of robots $n$. We can show that such algorithms cannot solve \visitall as
long as the number of steps is ${\cal O}(1)$.

\begin{theorem}\label{nstepvisitall}
There exists no $k$-step algorithm for \visitall, starting from any configuration $C_0 \in {\cal C}_{\odot}$, where $k={\cal O}(1)$. 
This result holds even if the system has chirality
\end{theorem} 

\begin{proof}
Let us consider a starting configuration $C_0 \in {\cal C}_{\odot}$ such that  robot $r_c$ is central, and all the robots \{$r_1,r_2,\ldots,r_{n-1}$\} that are not central are positioned in the vertices of a regular $n-1$-gon. 
Moreover, let us assume that the coordinate systems of robots $r_1,\ldots,r_{n-1}$ are pairwise different. The proof is by contradiction.
Let ${\cal A}$ be a \visitall algorithm that uses $k$-step. 
The first observation  is that any algorithm solving \visitall has to ensure that, for any robot $r_j$, there exists a configuration $C_j \in \Pi(C_0)$ that will be reached by the algorithm and such that $r_j$ occupies the central position in $C_j$, let this position be $p_c$. 
The second observation is that from a configuration in ${\cal C}_{\odot}$, algorithm ${\cal A}$ has to reach a configuration $C \not\in {\cal C}_{\odot}$ where the robot in position $p_c$ moved (see Lemma \ref{lm:sequence}).

The two above observations imply that the run of ${\cal A}$ that starts from configuration $C_0$  contains at least $n-1$ different configurations not in ${\cal C}_{\odot}$, one for each different robot in $\{r_1,r_2,\ldots,r_{n-1}\}$. This is immediate observing that from each $C_j$ the successive configuration in the run will be a specific $N_{j}$ different from the others $N_{i \neq j}$: starting from $C_j$ robot $r_{j}$ occupying $p_c$ has to move in a position different from the one where $r_{i \neq j}$ in configuration $C_{i}$ moved, this is ensured by the fact that the coordinate system of $r_j$ is different from the ones of other robots.

It is also easy to see that each $N_j$ could be a starting configuration for algorithm ${\cal A}$, thus   algorithm ${\cal A}$  starting from $N_j$ generates a run ${\cal R}_{{\cal A},N_j}:(X_0,X_1,X_2,\ldots)$ where permutations of $N_j$ appear at most every $k$-steps (this property is due to the fact that ${\cal A}$ is a $k$-step algorithm). Let call this fact (F1).  

Being the algorithm for oblivious robots, if $N_j$ appears in ${\cal R}_{{\cal A},C_0}$ at round $r$, then, from $r$ on, we have  ${\cal R}_{{\cal A},C_0}[r+t] = {\cal R}_{{\cal A},N_j}[t]$ for any $t>0$. Let this be fact (F2). 
Let $r^*$ the first round when, in  ${\cal R}_{{\cal A},C_0}$  every $N_j$, with $j \in [1,n-1]$, appeared. Let us consider the set $X$ of configurations that appears in ${\cal R}_{{\cal A},C_0}$ between round $r^*+1$ and $r^*+k$, clearly we have $|X| \leq k$. At the same time, all configurations $N_{*}$ appearing before $r^{*}-1$ have to appear at least once in the interval ${\cal R}_{{\cal A},C_0}[r^*+1,r^*+k]$, this comes directly from (F1) and (F2). By the pigeonhole principle, this is impossible: there are at least $n-1$ such configurations and $|X| \leq k < n-1$. 
This contradicts the existence of ${\cal A}$.
\end{proof}

\section{Oblivious Robots without Chirality}\label{nochir}

In this section we consider robots that do not share the same handedness. 
%
%
%
%
Interestingly, the absence of chirality changes the condition for solvability of \moveall and \visitall, showing the difference between these two problems. This is due to the fact that in systems without chirality, the configuration of robots may have mirror symmetry, in addition to rotational symmetry as in the previous section.  

\subsection{\moveall}\label{nochir:moveall}
The following theorem  illustrates the configurations for which the   \moveall    problem is unsolvable. 

\begin{theorem}\label{tmimp2}
In systems without chirality  \moveall is unsolvable in $1$-step starting  from any configuration $C_0 \in {\cal C}_{\odot}$,  as well as from any   configuration that  has a symmetry axis containing exactly one robot. 
\end{theorem} 
\begin{proof}
Unsolvability from $C_0 \in {\cal C}_{\odot}$ follows from Theorem~\ref{tm3}.
Consider then an initial  configuration $C_0$  such that in $C_0$ there exists a symmetry axis containing only one robot $r$, and let ${\cal A}$ be a solution algorithm for \moveall for this scenario.

In order to satisfy the \moveall specification, there must be a configuration $C_1 \in {\cal R}_{{\cal A},C_0}$ following $C_0$, and in $C_1$ a robot $r' \neq r$ must be in the position of robot $r$ in $C_0$; notice that if is not the case then $C_1 \not\in \Pi(C_0)$ or $r$ did not move, in both cases \moveall  specification has been violated.
However, in $C_0$ there exists a robot $r''$, symmetric to $r'$ with respect to the axis containing robot $r$,  that has the same view of $r'$,  this comes directly from the definition of symmetry axis. Now, if in $C_0$ robot $r'$ decides as destination the position of $r$, also $r''$ does the same, this implies that $C_1 \not\in \Pi(C_0)$. See the example of Figure~\ref{fig:noalgsymax}. This contradiction disproves the correctness of ${\cal A}$.\end{proof}

\begin{figure}[H]
  \centering
    \includegraphics[width=0.3\textwidth]{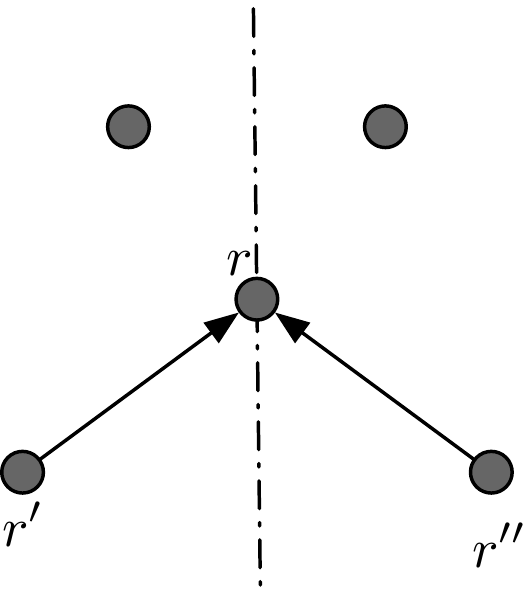}
      \caption{Initial configuration with an unique robot $r$ on the symmetry axis. If robot $r'$ moves in the position of $r$, then $r''$ does the same. }
      \label{fig:noalgsymax}
\end{figure}

We now consider the solutions to the \moveall problem for the feasible instances. 
If the configuration has a central symmetry (i.e., a rotational symmetry with $\theta=\pi$), each robot can be paired to its counterpart on the opposite end of the center, and the paired robots can swap positions.  (see Figure~\ref{fig:centralsymmetric}).
 
 When the initial configuration has no rotational symmetry nor any symmetry axes, then the robots can agree on a common chirality and the algorithms from the previous section can be applied (see Figure~\ref{fig:chirale1}).

Thus the only remaining configurations are those with an axis of symmetry. For such configurations, it is possible to partition the robots in three disjoint subsets, and it make them move as follows: (see also Figure~\ref{fig:symm1})
\begin{itemize}
\item[(i)] For the robots located on a symmetry axis, there exists a unique cyclic order on these robots. Robots on the axis are permuted according to this ordering. 
\item[(ii)] The second subset contains robots that are closer to one symmetry axis compared to other axes. These robots swap positions pairwise, each robot switching with its symmetric robot w.r.t. the closest axis. 
\item[(iii)] The last subset consists of robots that are equidistant from two distinct symmetry axes. Also in this case robots switch positions pairwise, and each one switches position with its symmetric robot w.r.t. the centroid $c$ of configuration $C_0$.
\end{itemize}

For all the configurations excluded by Theorem~\ref{tmimp2},   \moveall can be solved using the above approach. This algorithm is formally presented in Algorithm~\ref{alg:moveallnochi}. 

\begin{algorithm}
\caption{\moveall  $1$-step Algorithm when $C_0 \not\in {\cal C}_{\odot}$ and $C_0$ does not have a symmetry axis containing only one robot. \label{alg:moveallnochi}}
\footnotesize
\begin{algorithmic}[1]
\If{$C$ is central symmetric}\label{alg:moveallnochi:cc}

\State set {\bf destination} as position $p$, where $p$ is the symmetric of my position with respect to the center of $C$. \label{alg:moveallnochi:cc1}
\ElsIf{$C$ does not have symmetry axes} \label{alg:moveallnochi:nosim}
\State Set your clockwise orientation using a deterministic algorithm with input $C$.
\State Execute Algorithm~\ref{visitciclyc} of Section~\ref{1stepchirality} using the {\sc Order} procedure of Algorithm~\ref{alg:visitallrho2} of Section~\ref{1stepchirality}.
\Else
\If{My position is on a symmetry axis $A$}  \label{alg:moveallnochi:simaxis}
\State  Deterministically order the positions on $A$, obtaining a cylic order $p_0,p_1,\ldots,p_{t-Compute1}$. 
\If{My position is $p_i$}
\State set {\bf destination} as  $p_{(i+1)\mod t}$.
\EndIf
\ElsIf{There is an unique axis of symmetry $A$ closest to my position}

\State set {\bf destination} as position $p$ symmetric to my position with respect to $A$.
\Else
\State set {\bf destination} as position $p$ symmetric to my position with respect to the center of $C$. \label{alg:moveallnochi:simaxisend}
\EndIf
\EndIf

\end{algorithmic}
\end{algorithm}

\begin{theorem}\label{tmchirPOSS}
If $C_0 \not\in {\cal C}_{\odot}$ and $C_0$ does not have a symmetry axis containing exactly one robot, then \moveall is solvable in $1$-step even when the system does not have chirality. 
\end{theorem} 
\begin{proof}
The pseudocode is in Algorithm~\ref{alg:moveallnochi}. Let us consider an execution starting from a configuration $C_0$ that respects the assumptions of the theorem. 
\begin{itemize}
\item If in $C_0$ there is a central symmetry (see Figure~\ref{fig:centralsymmetric}) then Lines~\ref{alg:moveallnochi:cc}-\ref{alg:moveallnochi:cc1} are executed. By definition of central symmetric it is easy to see that the next configuration belongs to $\Pi(C_0)$ and that everyone has moved (e.g., see the arrows in Figure~\ref{fig:centralsymmetric}).

\begin{figure}
  \centering
    \includegraphics[width=0.6\textwidth]{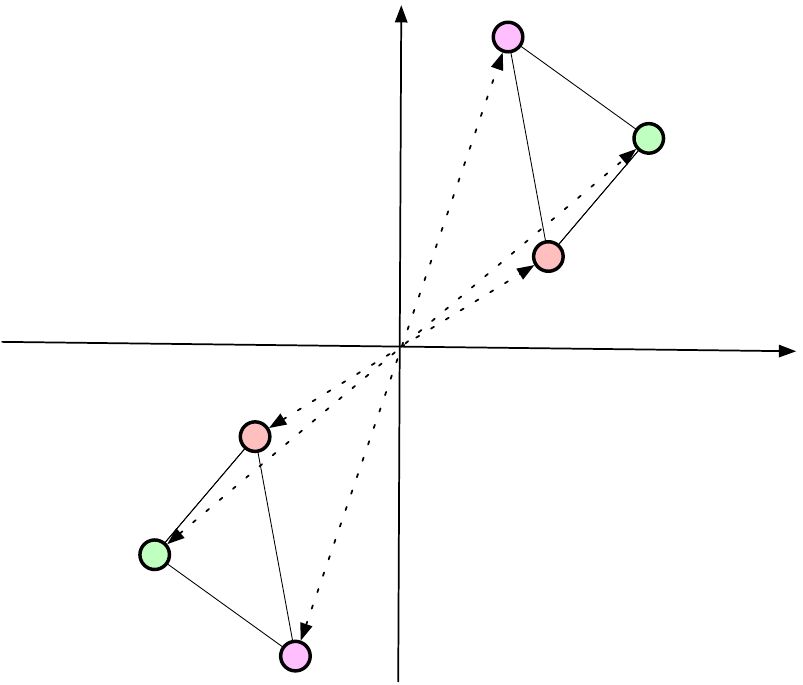}
      \caption{Central symmetric configuration that does not have a symmetry axis: each robot swaps position with the symmetric obtained by a rotation of $\pi$ radians with respect to the center.  Recall that a configuration $C$ has a central symmetry if once rotated around the centroid of $\pi$ radians we get a permutation of $C$.}
      \label{fig:centralsymmetric}
\end{figure}

\begin{figure}
  \centering
    \includegraphics[width=0.4\textwidth]{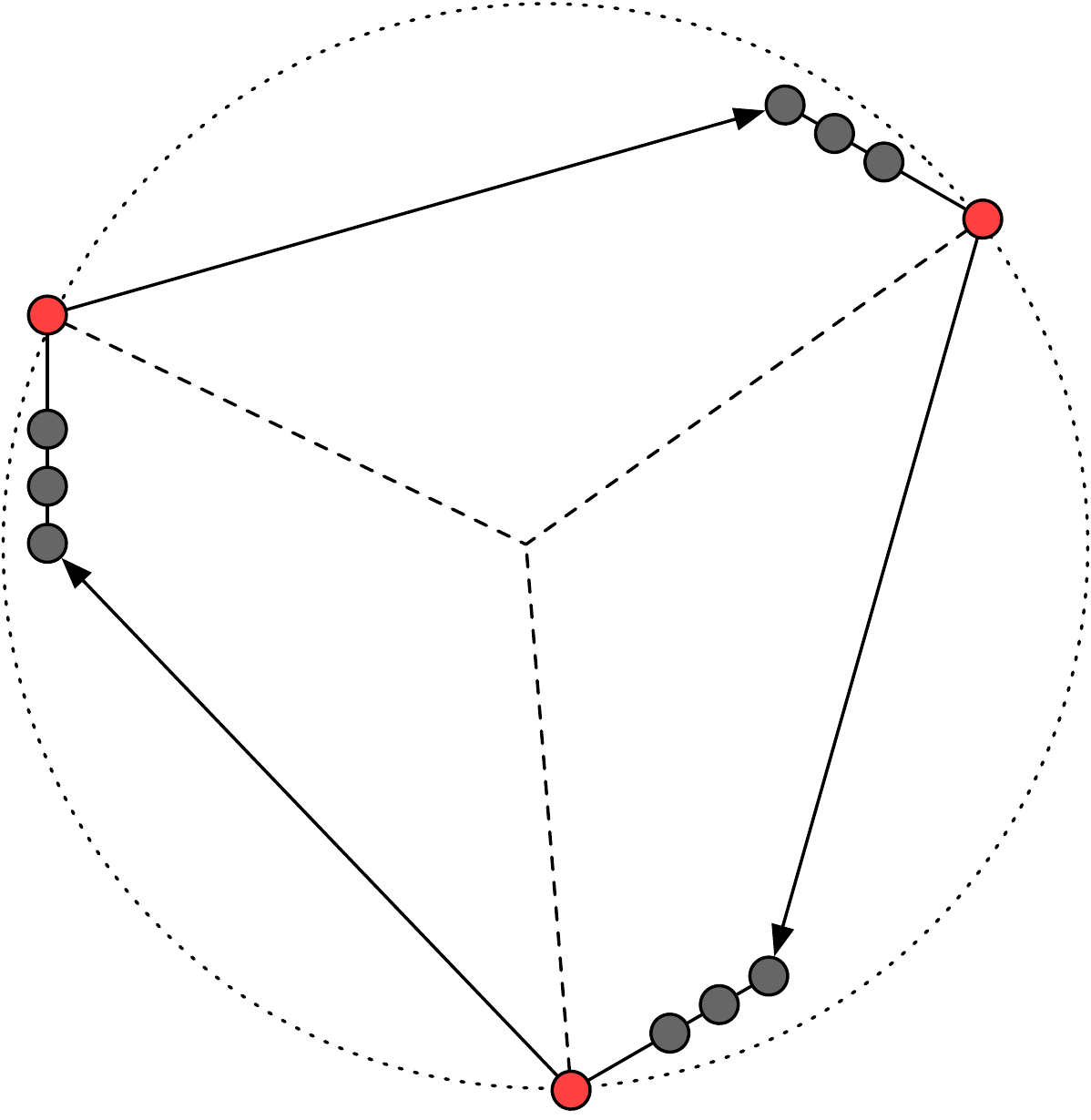}
      \caption{Configuration $C_0$ that has no central symmetry, no symmetry axis and where $\rho(C_0)=3$. Notice that robots can agree on a clockwise direction, see arrows. }
      \label{fig:chirale1}
\end{figure}

\item If $C_0$ has no central symmetry and there are no axes of symmetry (see an example in  Figure~\ref{fig:chirale1}), then the branch at line~\ref{alg:moveallnochi:nosim} is executed. 
If $\rho(C_0)=1$ there exists a total order among robots, thus it is possible to apply the  algorithms of Section~\ref{obv:chiral}. So let us examine the case when $\rho(C_0) \geq 2$.
The absence of a symmetry axis implies that robots can agree on a common notion of clockwise direction using configuration $C_0$.  Once they have a common clockwise notion they solve \moveall using the algorithms of Section~\ref{obv:chiral}. An example when $\rho(C_0) >1$ is shown in  Figure~\ref{fig:chirale1}.

\begin{figure}[H]
  \centering
    \includegraphics[width=0.5\textwidth]{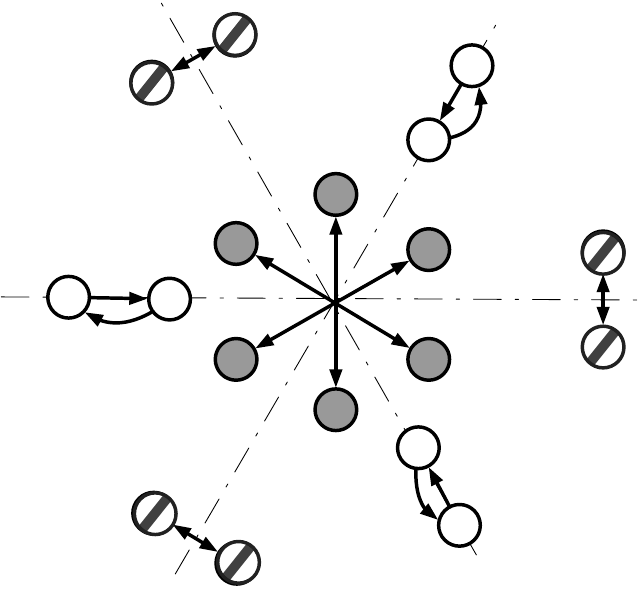}
      \caption{\small Configuration $C_0$ having 3 axes of symmetry. The arrows indicate three types of robot swaps: (1) Robots on the axis agree on a cyclic order (white robots); (2) Robots that are closer to one axis of symmetry swap positions w.r.t this axis (stripped robots); (3) Robots that are equidistant from two axes, swap position with symmetric robots w.r.t. the center (black robots). }
      \label{fig:symm1}
\end{figure}

\item If there is no central symmetry in  $C_0$ and each symmetry axis contains no robots or at least $2$ (see Figure~\ref{fig:symm1}), then we have 3 possible behaviours according to the position of the  robots.
\begin{itemize}

\item If a robot is on a symmetry axis $A$, it executes the branch at line~\ref{alg:moveallnochi:simaxis}. Robots on $A$ are able to agree on a cyclic order between them. Let us suppose the contrary, then there must exists a second symmetry axis $A'$ that is perpendicular to $A$, which contradicts the fact that $C_0$ has no central symmetry. Thus robots on $A$ can be cyclically ordered.  
It is immediate that if robots on $A$ permute according to this order, then the \moveall problem, restricted to the subset of robots on $A$, is correctly solved. 

\item If there exists a unique axis $A$ closest to a robot $r$, then the robot $r'$   symmetric to $r$ with respect to $A$ is unique and properly defined, and these two robots are able to swap position. Thus \moveall it is also solved correctly for this subset. 

\item If it does not exist a  unique axis $A$ closest to $r$ and $r$ is not on a symmetry axis, then we will show that exists a  unique  robot $r'$ symmetric to $r$ with respect to the center of   configuration $C_0$.
 Let $A$ and $A'$ be the two axes from which $r$ is equidistant. First, we observe that $r$ is on the bisector segment of the angle $\theta$ between two axes of symmetry. Let $B$ be this segment. Second, we observe that due to the fact that  $A$ and $A'$ are symmetry axes of $C_0$, we have that  $2\pi$ is divided by $\theta$, let us say $k$ times, and that there are $k$ rotations of $B$ around the center of $C_0$. This implies that there exists a  unique robot $r'$ (see Figure~\ref{fig:symm1} as an example where $\theta=\frac{\pi}{3}$). Thus $r$ and $r'$ can swap position safely.  Therefore, \moveall it is  solved correctly for this last subset. 
\end{itemize}
Since \moveall is correctly solved for all three cases, and each robot in $C_0$ belongs to one of the sets considered in the above cases, then \moveall is correctly solved.
\end{itemize}
\end{proof}

\noindent To summarize, we have the following characterization for solvability of \moveall without chirality:

\begin{theorem}\label{th:charmoveallnochi} In systems without chirality,
\moveall is solvable in $1$-step if and only if $C_0 \not\in {\cal C}_{\odot}$ and $C_0$ does not have a symmetry axis containing exactly one robot.
\end{theorem} 

\subsection{\visitall} 
\label{nochir:visitall}
The \visitall problem differs from \moveall only when $n > 2$, so we will assume in this section that $n \geq 3$. 
We will show that \visitall is solvable without chirality if (i) $C_0 \not\in {\cal C}_{\odot}$ and (ii) $C_0$ does not have symmetry axes, 
or there is a unique axis of symmetry that does not intersect any point of $C_0$. 
The main idea of the algorithm is the following. 
When $C_0$ does not have a symmetry axis: then it is possible to agree on a common notion of clockwise direction. Once this is done  Algorithm \ref{alg:visitallrho2} can be used. 
So we consider the case when $C_0$ has a  unique axis of symmetry that does not intersect any point of $C_0$: we partition $C_0$ in two sets $C'$ and $C''$, containing robots from the two sides of the axis of symmetry. In each of these sets it  is possible to agree on a total order of the points (recall that the symmetry axis is unique). Let $order':[p'_0,p'_1,\ldots,p'_{\frac{n-1}{2}}]$ be the order on $C'$
and $order'':[p''_0,p''_1,\ldots,p''_{\frac{n-1}{2}}]$ be the analogous for $C''$. We obtain a cyclic order on $C_0$ by having element $p''_0$ following $p'_{\frac{n-1}{2}}$, and, in a symmetric way,
$p'_0$ following $p''_{\frac{n-1}{2}}$ (see Figure \ref{fig:axisnoint}).

\begin{figure}[H]
  \centering
    \includegraphics[width=0.4\textwidth]{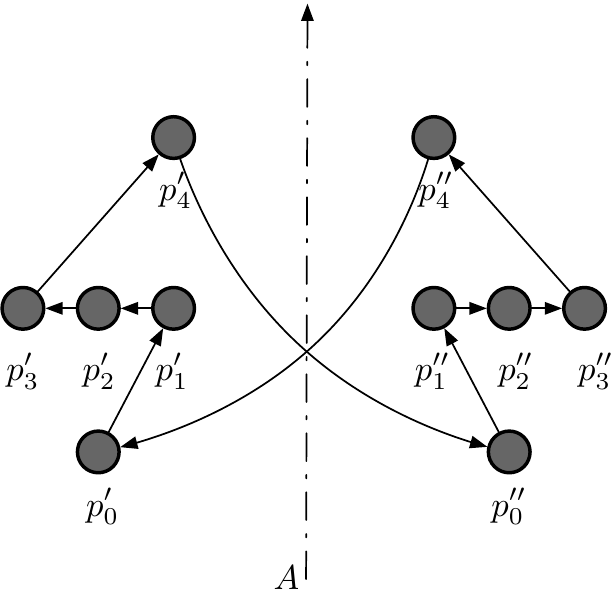}
      \caption{ Configuration $C_0$ with a unique symmetry axis $A$  and  no robots  intersecting $A$. 
      The arrow on axis $A$ indicate  the direction  on which robots agree. The arrows among configuration points indicate the cyclic order induced by Algorithm \ref{alg:visitallnochiralirty}.  }\label{fig:axisnoint}
\end{figure}

The corresponding ordering pseudocode is defined in Algorithm \ref{alg:visitallnochiralirty}. 
\begin{algorithm}[h]
\caption{{\sc Order} algorithm without Chirality. \label{alg:visitallnochiralirty}}
\footnotesize
\begin{algorithmic}[1]
\Procedure{Order}{Configuration $C$}
\If{$C$ does not have a symmetry axis}
\State Set your clockwise orientation using a deterministic algorithm with input $C$. \label{agreclockwise}
\State {\bf return} the output of  {\sc Order} procedure of Algorithm~\ref{alg:visitallrho2} of Section~\ref{1stepchirality}. \label{ref:ret}
\Else\Comment{In this case there is a unique symmetry axis that does not intersect any point of $C$.} \label{othercase}
\State Let $A$ be the unique symmetry axis of $C$. 
\State Let $C'$ and $C''$ be partitions of $C$ such that $C''$ is the symmetric of $C'$ w.r.t. $A$. 
\State Let $order'$ be a total order of robots in $C'$
\State Let $order''$ be a total order of robots in $C''$
\State Let $list$ be a list obtained by ordinately appending to $order'$ the elements of $order''$ 
\State {\bf return} list
\EndIf
\EndProcedure
\end{algorithmic}
\end{algorithm}

The correctness of Algorithm \ref{alg:visitallnochiralirty} is shown in the following theorem.

\begin{theorem}\label{th:charvisitallnochi}
When $n>2$ and robots do not have chirality, \visitall is solvable in $1$-step if the initial configuration $C_0 \not\in {\cal C}_{\odot}$ and one of the following holds:
\begin{enumerate}
\item There are no symmetry axes in $C_0$, or,
\item There exists a unique symmetry axis of $C_0$ and no point of $C_0$ intersects the axis.
\end{enumerate}
\end{theorem} 

\begin{proof}
 Algorithm \ref{alg:visitallnochiralirty} solves the problem. We show its correctness case by case:
\begin{itemize}
\item There are no symmetry axes in $C_0$: If $\rho(C_0)=1$, then it is obvious that the robots can agree on a common notion of clockwise direction, and in such a case the order returned at Line \ref{ref:ret} is correct by Lemma \ref{lemma:l1}. 
In case $\rho(C_0) >2$  it is possible to agree on a common notion of  clockwise direction (Line \ref{agreclockwise}) as follows. 
Since there exist an ordering among the classes constituting $C_0$, let $H$ be  the ``highest" class in this order  (see red points in Figure \ref{fig:chirale1} as reference). 
Take a line $L$ passing through the centroid $c$ of $C_0$ and a point $p$ in $H$. Since $L$ is not an axis of symmetry, we have that the half-planes defined by $L$ are different and not symmetric w.r.t. to $L$. Specifically, it is possible to order them using $p$ by scanning the plane starting from $p$ in the local clockwise and counter-clockwise directions using $c$ as center and stopping at the first asymmetry. 
 Therefore, the half-planes can be used to define a common clockwise direction, e.g., the one going from the lowest half-plane to the highest. Note that this procedure gives the same order for any choice of $p \in H$. 
 Once there is an agreement on clockwise direction the correctness derives from the one of Algorithm \ref{alg:visitallrho2} (see Lemma \ref{lemma:l1}).

\item There exists a unique symmetry axis $A$ of $C_0$ and no point of $C_0$ intersects $A$: this case starts at Line \ref{othercase}. Let $C'$ and $C''$ be partitions of $C_0$ such that $C''$ is the symmetric of $C'$ w.r.t. $A$. 
First of all, we show that robots can agree on a total order on points in $C'$: given the axis $A$ they can agree on an orientation of the axis (suppose the contrary, then there should exists a second axis intersecting $A$, that is excluded by hypothesis). 
We order the robots using the coordinates on $A$ from lowest to highest (robots with the same coordinate are ordered by their distance from the axis). Let this order be $order'$.
 The analogous can be done for $C''$, let this order be $order''$. An example is in Figure \ref{fig:axisnoint}. 
We merge these two orders by taking the first point in $order''$, let it be $p''_0$,  the last point in  $order'$, let it be $p'_{\frac{n-1}{2}}$, and by creating a common cyclic order on which $p''_0$  follows $p'_{\frac{n-1}{2}}$. Symmetrically,
the first point   $p'_0$   in $order'$  follows the last point $p''_{\frac{n-1}{2}}$ in $order''$. 

\end{itemize}
\end{proof}

Interestingly, without chirality, \visitall is not solvable if the assumptions of Th.~\ref{th:charvisitallnochi} do not hold: 
\begin{theorem}\label{th:charvisitallnochin}
When $n>2$ and there is no chirality, there exists no algorithm that solves  \visitall in $1$-step from an initial configuration $C_0$ if one of the following holds:
\begin{itemize}
\item $C_0 \in {\cal C}_{\odot}$
\item There exists a symmetry axis of $C_0$ intersecting a proper non-empty subset of $C_0$. 
\item There are at least two symmetry axes of $C_0$. 
\end{itemize}
\end{theorem}

\begin{proof}

\begin{figure}
  \centering
    \includegraphics[width=0.6\textwidth]{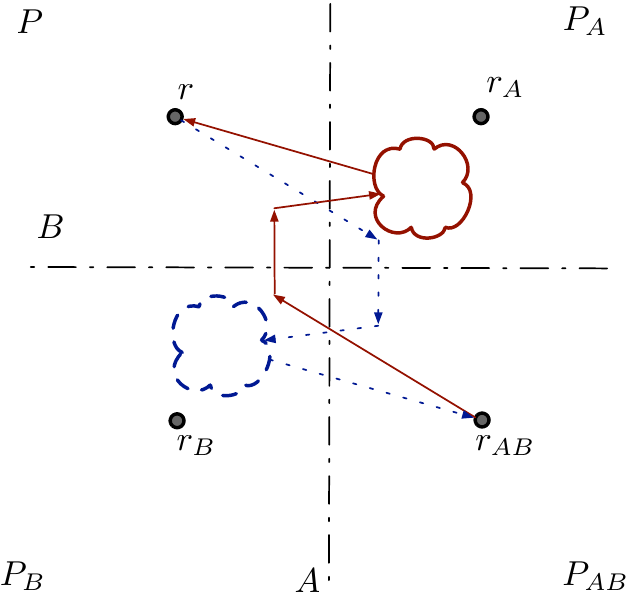}
      \caption{Two symmetry axes in $C_0$ and no chirality.
      }\label{twoaxesimpossibility}
\end{figure}
We prove the impossibility case by case:

\begin{itemize}
\item $C_0 \in {\cal C}_{\odot}$: this case derives directly from Theorem \ref{tm3}. 
\item  There exists a symmetry axis $A$ in $C_0$ and $A$ intersects a proper non-empty subset of $C_0$: to solve \visitall we must have that a robot $r$ outside of  axis $A$  eventually moves to a position on the axis, let such a position be $p$. However, being $A$ a symmetry axis, we must have a robot $r'$ symmetric to  $r$ w.r.t. $A$. Robot $r'$ also moves to position $p$. Since both robots move to point $p$ we reach a configuration that is not in $\Pi(C_0)$.

\item There are at least two symmetry axes in $C_0$: let $A$ and $B$ be these axes. Let $P, P_{A},P_{B},P_{AB}$ be the portions of the plane defined by the axes, see Figure \ref{twoaxesimpossibility}. We can assume that no point in $C_0$ intersects the axis (otherwise we are in the previous case). 
Let $r$ be a robot in $P$, and let $r_{A}, r_{B}, r_{AB}$ be the symmetric robots with respect to axis $A$, $B$ and both axes (see always Figure \ref{twoaxesimpossibility} for a reference), and let $p,p_{A},p_{B},p_{AB}$ be their respective positions in $C_0$.  Let us suppose, by contradiction, that there exists an algorithm ${\cal A}$ solving \visitall on $C_0$,
and let ${\cal R}_{{\cal A},C_0}:(C_0,C_1,C_2,\ldots)$ be the corresponding run. 
We examine the path of robot $r$ among the points in $C_0$ in the run ${\cal R}_{{\cal A},C_0}:(C_0,C_1,C_2,\ldots)$. Let us suppose, w.l.o.g., that in $C_1$ robot $r$ moves to a point in $P_{A}$ (see the blue dotted path in the Figure \ref{twoaxesimpossibility}). 
It is obvious that the local coordinate systems of the robots could be such  that: robot $r_{A}$ will do a symmetric move to a point in $P$, robot $r_{AB}$ a symmetric move to $P_{B}$ (see the red path in Figure \ref{twoaxesimpossibility}), and so on. 
Let $t$ be the first round at which $r$ moves to a location in $\{p_{A},p_{B},p_{AB}\}$. Since ${\cal A}$ solves \visitall such a  round $t$ must exist, and w.l.o.g let us suppose that $r$ visits first the location $p_{AB}$. By symmetry considerations we have
that, at round $t$, also robot $r_{AB}$ has to move to location $p$ (recall that $p$ is the location occupied by $r$ in $C_0$).
 Therefore, from round $0$ to $t$ robots $r$ and $r_{AB}$ have walked on a closed path, let it be $CP$, between some locations of $C_0$ that includes $p$ and $p_{AB}$ but which does not include  $p_A,p_B$ (this comes directly from the definition of $t$). The $CP$ path is represented in Figure \ref{twoaxesimpossibility} as the union of the red path, the blue paths and the clouds; clouds represent the move made by $r$ and $r_{AB}$ between rounds $3$ and $t$.  Since ${\cal A}$ is an algorithm for oblivious robot, this implies that robots $r$ and $r_{AB}$ will only visit  points in $CP$  also in future rounds. Therefore, ${\cal A}$ cannot be correct. 

\end{itemize}
\end{proof} 

To summarize, we have the following:

\begin{theorem}\label{th:charvisitallnochi-iff} In systems without chirality,
\visitall is solvable in $1$-step if and only if 
$C_0 \not\in {\cal C}_{\odot}$ and either  there are no symmetry axes in $C_0$, or 
there exists a unique symmetry axis  that does not intersect any  point  of $C_0$.
\end{theorem} 

\subsection{ $2$-step Algorithms}\label{nochir:2step}
In this section we show that using $2$-step algorithms does not help in enlarging the class of solvable configurations.
 
\subsubsection{ \moveall}
Obviously, Theorem ~\ref{tmimp1} holds also when there is no chirality. One may wonder whether is possible to solve \moveall with a $2$-step algorithm when $C_0$ has a symmetry axis with a unique robot on it. However, the argument used in Theorem~\ref{tmimp1}  can be adapted also for this case showing that this is impossible.

\begin{theorem}\label{lemmax} When there exists an axis of symmetry  in $C_0$ containing a single robot and there is no chirality, then there exists no $2$-steps algorithm  that solves \moveall.
\end{theorem} 

\begin{proof} 
Consider a configuration $C_0$ and let $A$ be an axis with a single robot on it. Let us assume that  algorithm ${\cal A}$ is a correct $2$-step algorithm that solves \moveall starting from $C_0$, and let ${\cal R}_{{\cal A},C_0}:(C_0,C_1,C_2,C_3,\ldots)$.
Let $r$ be the robot in $A$ and let $p$ be its position.  In configuration $C_1$ robot $r$ has to move to break the symmetry in such a way that in configuration $C_2$ another robot could substitute $r$ on $A$. Suppose the contrary, if $r$ remains on the axis in $C_1$ the all the other robots
will be symmetric. Therefore, in $C_2$ is not possible for a single robot, different than $r$, to reach position $p$.  Therefore let $p'$ be the position of $r$ in $C_1$. 

Let $r'$ be the robot that exchanges position with $r$ in configuration $C_2$. It is clear that if $r'$ has a different local coordinate system than $r$, then in configuration $C_3$ it will move to a position $p''\neq p'$. Therefore, we have $C_3 \not\in \Pi(C_1)$, violating  the \moveall specification (see Lemma \ref{lemma:prem}).
\end{proof}
 
From the previous theorem and Theorem~\ref{tmimp1}, we have:

\begin{theorem}\label{newimp} 
 When the system has no chirality,  
 \moveall is not solvable in $2$-steps, from an initial configuration $C_0$, if  
 $C_0 \in {\cal C}_{\odot}$, or if there exists an axis of symmetry in $C_0$  containing a single robot. 
\end{theorem} 

\subsubsection{ \visitall}


\begin{theorem}\label{th:2charvisitallnochi}
When $n>2$ and there is no chirality, \visitall is not solvable in $2$-steps, from an initial configuration $C_0$, if one of the following holds:
\begin{itemize}
\item $C_0 \in {\cal C}_{\odot}$
\item There exists a symmetry axis $A$ of $C_0$ intersecting a proper non-empty subset of $C_0$. 
\item There are at least two symmetry axes of $C_0$. 
\end{itemize}
\end{theorem}

\begin{proof} 
We prove the theorem case by case:
\begin{enumerate}
\item $C_0 \in {\cal C}_{\odot}$: this case derives directly from Theorem  \ref{tmimp1}. 
\item  There exists a symmetry axis in $C_0$ and it intersects a proper non-empty subset of $C_0$.\\
Let us assume that algorithm ${\cal A}$ is a correct $2$-step algorithm that solves \visitall starting from $C_0$, and let ${\cal R}_{{\cal A},C_0}:(C_0,C_1,C_2,C_3,\ldots)$. It is clear that, to break the symmetry, in each configuration $C_j$, with $j$ odd, the symmetry has to be broken by moving one of the robots located on the symmetry axis. Let  $p$ be the position of this robot in $C_{j-1}$, note that such a position is the always the same: the decision on which robot has to move is taken always on configurations that are permutations of $C_0$, thus it will   always move a robot to the same specific position. Let $r_x$ and $r_y$ be two robots with different local references systems. Since ${\cal A}$ solves \visitall there must exist a configuration $C_{j_x}$ where robot $r_x$ is in position $p$; analogously, there must exist a $C_{j_y}$ where robot $r_y$ is in position $p$.  
In configuration $C_{j_x+1}$ robot $r_x$ has to move outside of the axis. The same happens to $r_y$ in configuration $C_{j_y+1}$. However,  the coordinate systems could be such that $ C_{j_x+1} \not\in \Pi(C_{j_y+1})$, and by  Lemma \ref{lemma:prem}, this violate the specification of \visitall: $r_x$ and $r_y$ decides where to move by looking at exactly the same snapshot; therefore, they will move on opposite locations of the axis if their local coordinate systems have opposite chirality. 

\item There are at least two symmetry axes of $C_0$:  first of all we can assume that it does not exist an axis of $C_0$ containing robots, otherwise we boil down to previous case.  For any possible movements of robots, they have to keep the same symmetry. This implies that the argument of Thm. \ref{th:charvisitallnochin} still applies. 

\end{enumerate}
\end{proof}

\section{Oblivious Robots with Visible Coordinate Systems and Chirality}

In this section, we assume that each robot can see the coordinate system of all robots and the system has chirality. 
As we have seen in Section \ref{obv:chiral}, with chirality, the only configurations in which \visitall cannot be solved are the ones in ${\cal C}_{\odot}$. 
We now present a {\sc Voting} algorithm that solves \visitall also starting from  these configurations, provided that robots have this extra knowledge of the coordinate systems of other robots.
The algorithm (see Algorithm \ref{alg:moveallvisaxis}) uses Procedure {\sc innerPolygon}, which takes a configuration $C$ and   returns only the points on the smallest non degenerate circle having the same center as $SEC(C)$ and passing through at least one point of $C$ (e.g. see the white points in Figure \ref{f:vote1}).

When $C_0 \not\in {\cal C}_{\odot}$, the algorithm uses the {\sc Order} procedure from Section \ref{obv:chiral}. In case the initial configuration is in  ${\cal C}_{\odot}$,
the algorithm implements a voting procedure to elect a unique vertex of the innermost non-degenerate polygon $P$ computed by Procedure {\sc innerPolygon}.
(see, an example in Figure \ref{f:special3}).
The vote of a robot $r$ is computed by translating its coordinate system to the center of $SEC(C_0)$. The vote of $r$ will be given to the point of $P$ that forms the smallest counter-clockwise angle with the $x$-axis of the translated system. Since the number of robots is co-prime to the size of $P$, a unique vertex (robot) can be elected and the elected point is used to break the symmetry and compute a total order among the robots. As before, the robots use this total order to move cyclically solving \visitall.

\begin{algorithm}
\caption{ {\sc Order} Algorithm when robots have visible coordinate systems. \label{alg:moveallvisaxis}}

\begin{algorithmic} 

\Procedure{getVote}{Polygon $P$, robot $r$}
\State $o=getCenter(P)$
\State Le robot $p_v$ in $P$ be the robot that forms the smallest clockwise angle with the $x$-axis of the coordinate system $Z_r$ of robot $r$, when $Z_r$ is translated in $o$. 
\State {\bf return} $p_v$
\EndProcedure

\Procedure{Voting}{Configuration $C$}
\State $P=innerPolygon(C)$
\State $V$=vector of size $|P|$ with all entries equal to $0$.
\ForAll{$r \in C$}
\State $r_v=getVote(P,r)$
\State $V[v]=V[v]+1$
\EndFor
\State $p_l=$ elect one robot in $P$ using the votes in $V$.
\State {\bf return} $p_l$
\EndProcedure

\Procedure{Order}{Configuration $C$}
\If{$C \in {\cal C}_{\odot}$}
\State $p_l=Voting(C)$
\State Compute a cyclic order on positions in $C$ using the leader robot $p_l$.
\Else
\State Compute an order using {\sc Order}($C$) of Algorithm \ref{alg:visitallrho2} in Section \ref{1stepchirality}.
\EndIf
\State return the cyclic order computed
\EndProcedure
\end{algorithmic}
\end{algorithm}

\begin{figure}
\begin{center}

\begin{subfigure}{0.45\textwidth}
  \includegraphics[width=0.7\textwidth]{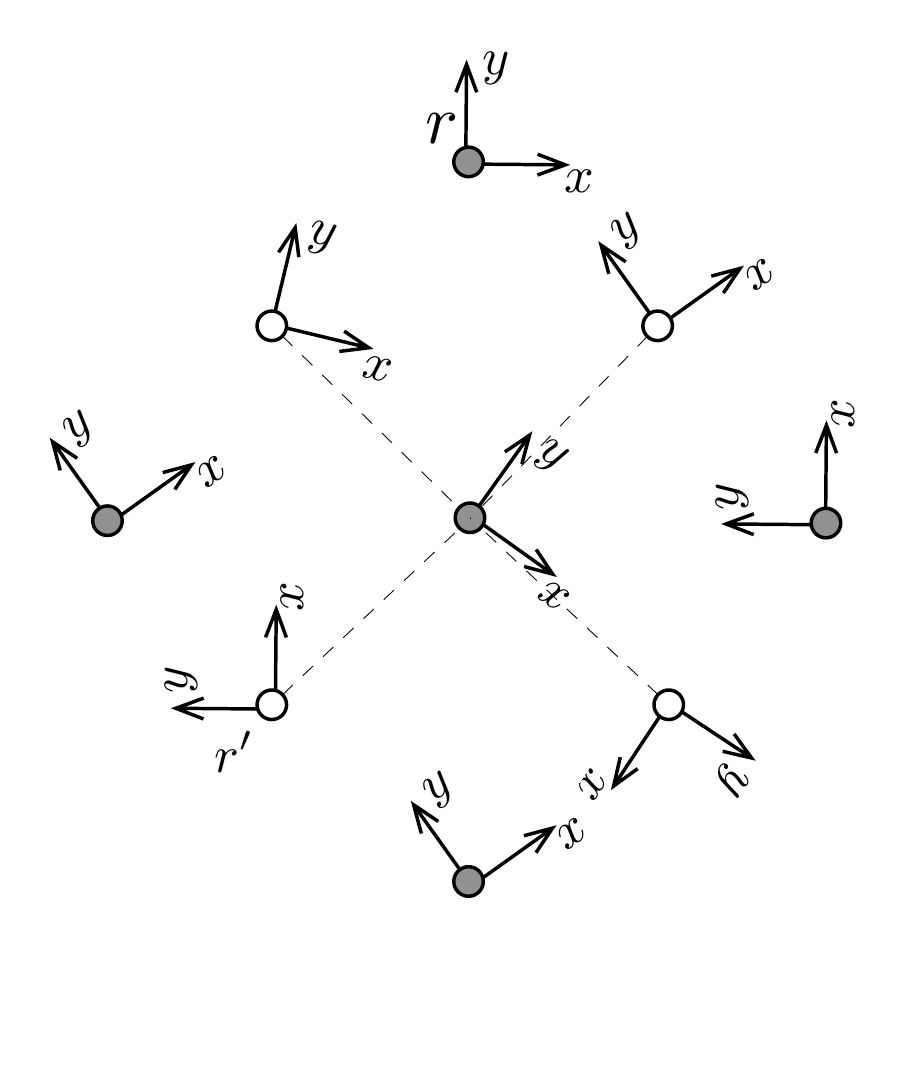}
  \caption{Initial configuration. The robots on the vertices of the innermost polygon $P$ are the white ones.}
  \label{f:vote1}
\end{subfigure}
$\,\,$
\begin{subfigure}{0.45\textwidth}
  
  \includegraphics[width=0.7\textwidth]{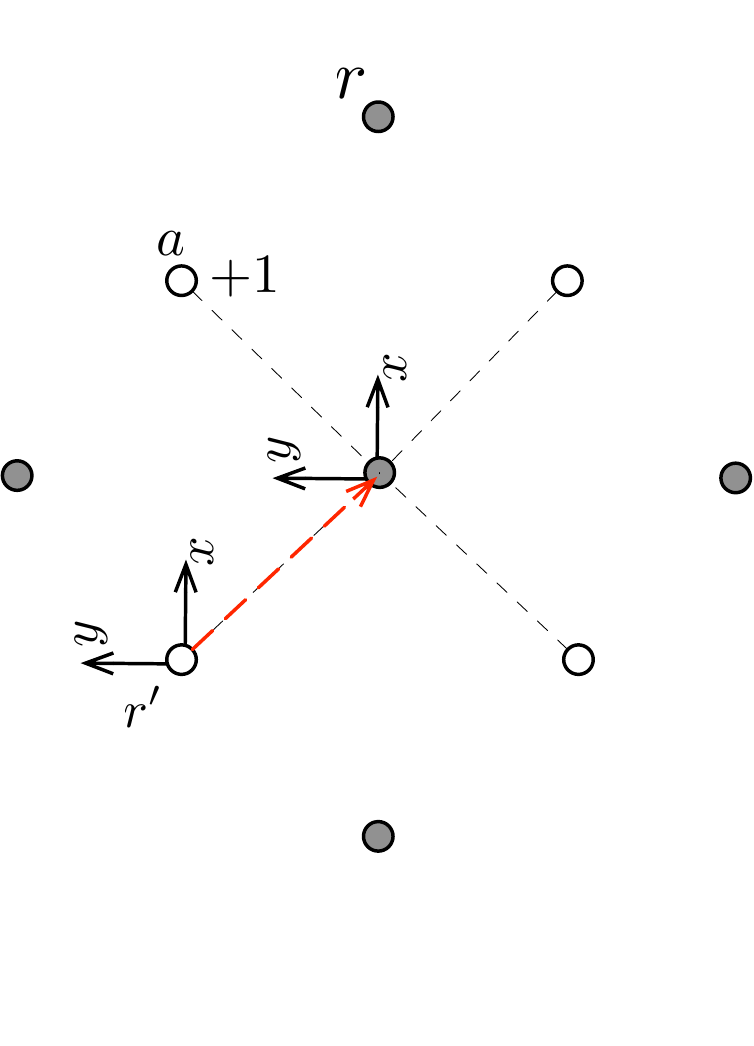}
  \label{f:vote2}
  \caption{Robot $r$ computes the vote of $r'$: it first translate the coordinate system of $r'$ in the central robot, then it assigns the vote to vertex $a$ in $P$. This vertices is voted since is the one, among all the other vertices of $P$, that forms the smallest counter-clockwise angle with the axis $x$ of $r'$.}
 
    \end{subfigure}%
    \\
\begin{subfigure}{0.3\textwidth}
 \includegraphics[width=\textwidth]{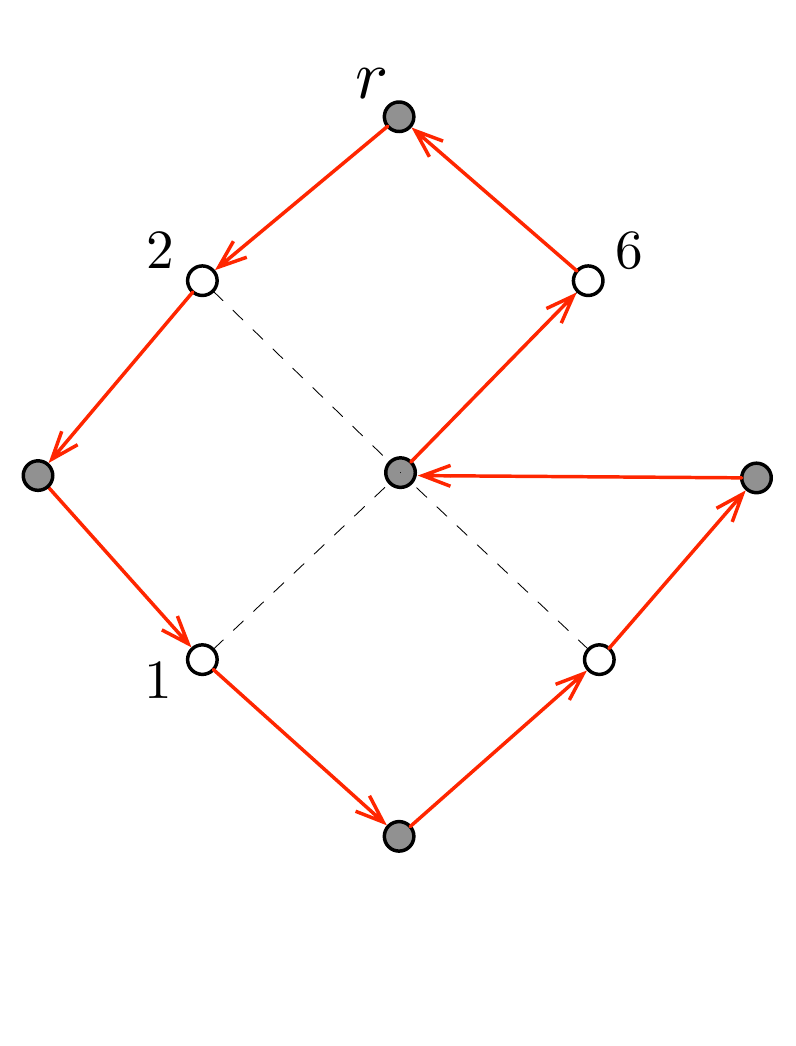}
  \label{f:vote3}
  \caption{Votes distribution and induced cyclic order.}
  \end{subfigure}%
\end{center}
\caption{Robots with visible coordinate system, voting procedure.}
\label{f:vote}
\end{figure}

\begin{theorem}\label{th:visaxis}
If each robot can see the axes of the others and there is chirality, then there exists a $1$-step algorithm solving \visitall for any initial configuration $C_0$.
\end{theorem} 
\begin{proof}
We have to prove the correctness of Algorithm \ref{alg:moveallvisaxis} only for configurations in ${\cal C}_{\odot}$. 
Let $k$ be the symmetricity of the initial configuration $C_0$ without the central robot, and let $P$ be the innermost non degenerate $\beta \cdot k$-gon in $C_0$ with $\beta \in \mathbb{N}^{+}$.

The key observation is that $n=(\alpha+\beta) \cdot k+1$ for some $\alpha \in \mathbb{N}^{+}$.
We now show that no matter which robot calls procedure {\sc Voting}, the procedure returns always the same point in $P$. 
{\sc Voting} iterates over all robots, and computes the vote of each robot $r \in C_0$.
The vote is computed by first translating  the coordinate system of $r$ to the center of configuration $C_0$ and then taking as voted robot $r_v$, the one that makes the smallest counter-clockwise angle with the $x$-axis of the translated coordinate system. 
An example of the voting procedure is in Figure \ref{f:vote}.
First of all note that the result of the {\sc Voting} procedure is independent from the robot that is executing it, and it always returns the same distribution of votes for points in $P$, even if we permute the robots in $C_0$. 

 We now show that the distribution of votes  cannot be symmetric, and one robot in $P$ can be elected. 
 
 The proof is by contradiction. Let us assume that there exists an axial symmetry on the distribution of votes. The symmetry axis may cross two robots,   one robot, or none. In case it crosses only one robot, then we elect that robot  as leader. If the axis crosses two robots or none, then the number of votes must be even and so the number of points in $P$, that is $\beta \cdot k$; but this is obviously impossible since $\beta \cdot k$ and $n$ are co-prime. 

Let us assume that there exists a rotational symmetry on the votes. Then, there exists a proper divisor $d>1$ of $\beta \cdot k$ and an ordering of the robots in $P$ such that $d \cdot (V[1]+V[2]+\ldots+V[\frac{\beta k}{d} ])=n$, where $V[j]$ is the number of votes for a point of $P$ in position $j$ of the aforementioned ordering. Notice that  this would imply that $n$ and $\beta k$ are not co-prime, which is a contradiction. 

It follows that  it is always possible to elect a leader in $P$ using the votes. 

Therefore, procedure {\sc Voting} returns the same leader robot $p \in P$ for any permutation of the robots in $C_0$. It is obvious that the presence of $p$ breaks the symmetry of the configuration and it allows to compute a total order among positions in $C_0$ shared by all robots. Once this total order is given the solution is immediate. \end{proof}

\section{Robots with one bit of Persistent Memory}
\label{sec:mem}

Motivated by the impossibility result of Theorem  \ref{tmimp1}, we investigate robots with some persistent memory. Interestingly, we show that a single bit of memory is sufficient to overcome the impossibility, and solve \visitall using a $2$-step algorithm.  
Note that we cannot overcome the impossibility using $1$-step algorithms, as Theorem~\ref{tm3} holds even if the robots are equipped with an infinite amount of memory. 

We present the $2$-step algorithm below (Algorithm \ref{alg:visitall1bit}) for $n\geq 3$ robots.

\begin{algorithm}
\caption{$2$-step \visitall with one bit of memory.  \label{alg:visitall1bit}}

\begin{algorithmic}[1]
\Procedure{Init}{}
\State $b=0$
\EndProcedure
\\
\State $C \leftarrow$ {\sc Look}
\\

\Procedure{Compute}{Configuration $C$}
\If{$C \not\in{\cal C}_{\odot} \land b=0$} \label{alg1bit:notcentral}
\State Compute an $order$ using Algorithm \ref{alg:visitallrho2} with input $C$.
\State Permute robots according to the computed $order$.   \label{alg1bit:notcentralend}
\ElsIf{$C \in{\cal C}_{\odot} \land b=0$}   \label{alg1bit:central}
\State b=1  \label{alg1bit:setbit}
\If{I am the central robot} \label{alg1bit:centralrobot}
\State Compute a destination point $\mathbf{v} = ${\sc ComputeMovementCentral}$(C)$. 
\State set {\bf destination} as  $\mathbf{v}$
\EndIf
\ElsIf{$C \in{\cal C}_{\odot} \land b=1$} \label{alg1bit:notcentralrobot}

\State compute a destination point $\mathbf{v} = ${\sc ComputeMovementNotCentral}$(C)$. 
\State set {\bf destination} as  $\mathbf{v}$ \label{alg1bit:notcentralrobotend}

\ElsIf{$C \not\in{\cal C}_{\odot} \land b=1$}  \label{alg1bit:wascentral}
\State $(C', p', Leader)=${\sc Reconstruct}$(C)$ \label{alg1bit:reconstruct}
\State Compute a cyclic order $p_0,p_1,\ldots,p_{n-1}$ of positions in $C'$ using the pivot robot $p'$.
\If{I am the {\em Leader}}
\State b=1 \label{alg1:leadernoreset}
\Else 
\State b=0 \label{alg1bit:unsetbit}
\EndIf
\If{my position in $C'$ was $p_i$}
\State set {\bf destination} as $p_{(i+1)\mod n}$ \label{alg1bit:wascentralend}
\EndIf

\EndIf
\EndProcedure
\\
\State {\sc Move:} to {\bf destination} 
\end{algorithmic}
\end{algorithm}

\bigskip
\noindent {\bf Intuitive description of the algorithm.}
The general idea is to use alternate rounds of communication and formation of the actual permutation.
In the communication round, the robots create a special intermediate configuration that provides a total order on the robots; In the subsequent round they reconstruct the initial pattern forming the permutation of the initial configuration. The memory bit is crucial to distinguish the intermediate configuration from the initial configuration. 
If the initial configuration $C_0 \not\in {\cal C}_{\odot}$ then the robots follow the 1-step algorithm described in Section \ref{obv:chiral} and we will show that this does not conflict with the rest of the algorithm designed for the case when $C_0 \in {\cal C}_{\odot}$, as described below.

Initially every robot has the bit  $b$ set to $0$. 
When a robot observes that the configuration is in ${\cal C}_{\odot}$ and bit $b$ is $0$,  it sets the bit to $1$ to remember that the initial configuration $C_0 \in {\cal C}_{\odot}$. 
The central robot $r_l$ in $C_0$ takes the role of $Leader$ and performs a special move to create the intermediate configuration $C_1$ that is  not in ${\cal C}_{\odot}$  but from
$C_1$, it is possible to reconstruct the initial configuration $C_0$ or any permutation of it (This move is determined by procedure {\sc ComputeMovementCentral} described in the next paragraph). 

A key point of the algorithm is that the $Leader$ robot remains invariant.  
At the next activation, the robots observe a configuration that is not in ${\cal C}_{\odot}$ and they have bit $b=1$; this indicates that this is an intermediate configuration and the robots move to reconstruct a configuration $C_2 = \Pi(C_0)$. 
With the exception of the Leader $r_l$ whose memory bit $b$ is always set to $1$,  all the other robots will now reset their bit $b$ to 0.

In the next round, the robots are in configuration $C_2$, where the central robot $r_c$ is not $r_l$ (the robots have performed one cyclic permutation). 
At this point, 
the robot $r_l$ is the unique robot whose bit $b=1$. All other robots have $b=1$ and they behave similarly as in the first round, including robot $r_c$ which moves like the central robot moved in $C_0$.
However, the leader robot $r_l$ also moves at the same time, in a special way (as described in procedure {\sc ComputeMovementNotCentral} presented in the next paragraph).
The combination of moves of the leader robot and the central robot allows the robots not only to reconstruct the initial configuration, but also to uniquely identify a ``pivot" point in the pattern (see Figure~\ref{fig:cmove}), which is kept invariant during the algorithm.  The reconstruction is executed by procedure  {\sc Reconstruct} whose details are discussed below.
The recognition of the pivot point allows the robots agree on the same cyclic ordering of the points in the initial pattern, thus allowing cyclic permutations of the robots.

\noindent A pictorial representation of the algorithm is presented in Figure \ref{fig:runcentral}. 
\begin{figure}[tbh]
\center
\begin{subfigure}{0.45\textwidth}
    \includegraphics[width=0.7\textwidth]{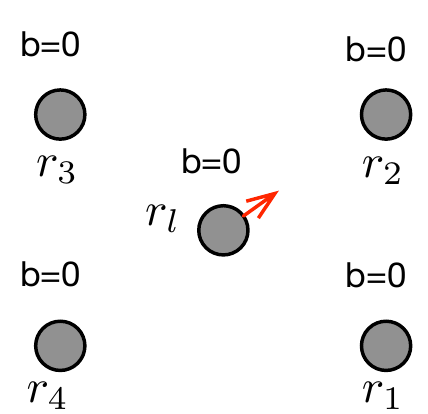}
  \caption{\footnotesize round $r=0$ ({\em communication round}):
  only the central agent $r_l$ moves.}
  \label{f:elements1}
\end{subfigure}
$\,\,$
\begin{subfigure}{0.45\textwidth}  
 \includegraphics[width=0.7\textwidth]{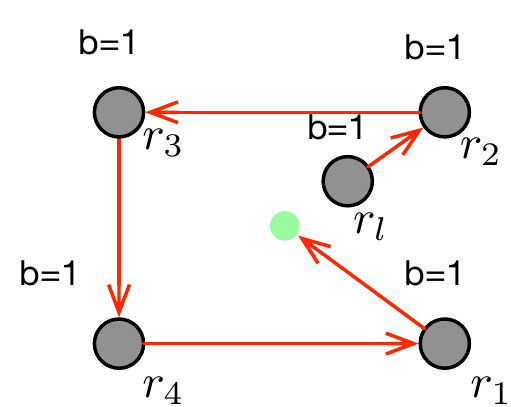}
    \label{f:elements2}
  \caption{\footnotesize round $r=1$  ({\em permutation round}).
  The arrows show the cyclic order over the positions of the original central configuration. 
  }
    \end{subfigure}%
    \\
\begin{subfigure}{0.45\textwidth}
\includegraphics[width=0.7\textwidth]{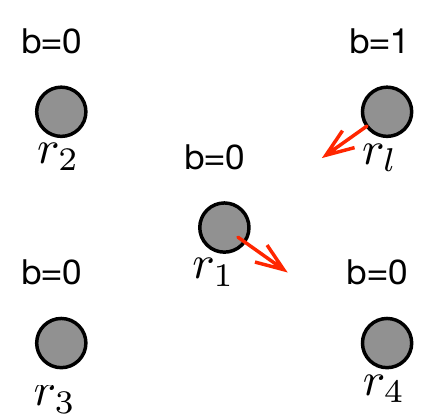}
    \label{f:elements3}
  \caption{\footnotesize round $r=2$  ({\em communication round}).
  $r_l$ moves since it has $b=1$ and $r_1$ moves because it is a central robot.}
  \end{subfigure}%
  $\,\,$
  \begin{subfigure}{0.45\textwidth}
\includegraphics[width=0.7\textwidth]{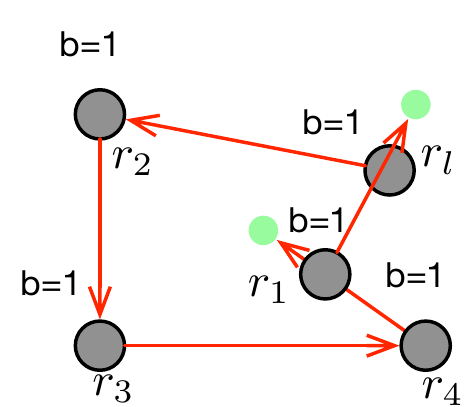}
  \label{f:elements4}
  \caption{\footnotesize round $r=3$  ({\em permutation round}). 
  The robots move according to the  cyclic order on the reconstructed central configuration. 
  }
  \end{subfigure}%

\caption{Case of central robot moving when $n>3$.}
\label{fig:runcentral}
\end{figure}

\bigskip
\noindent {\bf Movements of $r_l$ and $r_c$.} This paragraph discusses the implementation details of functions {\sc ComputeMovementCentral} and {\sc ComputeMovementNotCentral} used in Algorithm \ref{alg:visitall1bit}.  In particular, special care has to be taken to design the movements of the robot leader $r_l$   and of the central robot $r_c$,  if different from $r_l$.
Such movements have to be done in such a way to break the symmetry of the configuration by electing always the same pivot position $p'$, and to make the central configuration reconstructable.   Let $C$ be a generic central configuration in $\Pi(C_0)$.  Let $P_0,P_1,\ldots, P_m$ be a decomposition of $C$ in concentric circles, where each $P_j$ is a circle, $P_0$ is the degenerate circle constituted by the only central robot, and $P_1$ is the innermost non-degenerate polygon on which $p'$ resides (see Figure \ref{fig:conc}). 

\begin{figure}[tbh]
\center
  \includegraphics[width=0.4\textwidth]{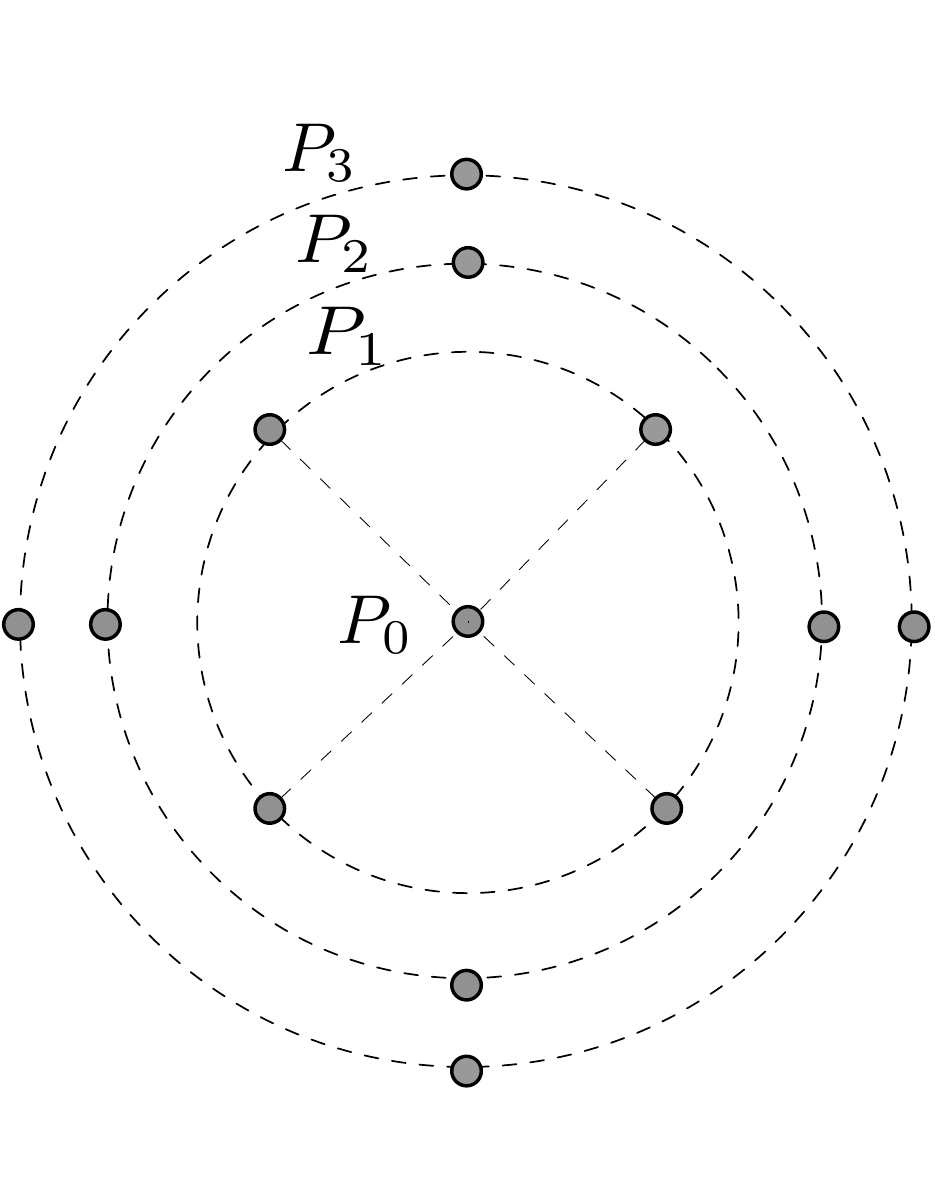}
  \caption{Decomposition of a central configuration in concentric circles. Note that $P_0$ is the degenerate circle formed by only the central robot.}\label{fig:conc}
\end{figure}

The function {\sc ComputeMovementCentral}, called by robot $r_c$, executes the following movements according to the number of robots  $n$:
\begin{itemize}
\item $n=3$  (Case C1):  In such a case, since  $C \in C_{\odot}$,  the  robots are on a single line.  Let $s$ be the segment containing all robots. Robot $r_c$ moves perpendicularly to $s$ of a distance $d=\frac{|s|}{2}$.  The direction is chosen  such that the pivot position $p'$ will be the one of the first robot encountered travelling along the arcs that connect the endpoints of $s$ and $r_c$ in the clockwise direction (see Figure \ref{f:elements1}). 
\begin{figure} [tbh]
\center
  \includegraphics[width=0.4\textwidth]{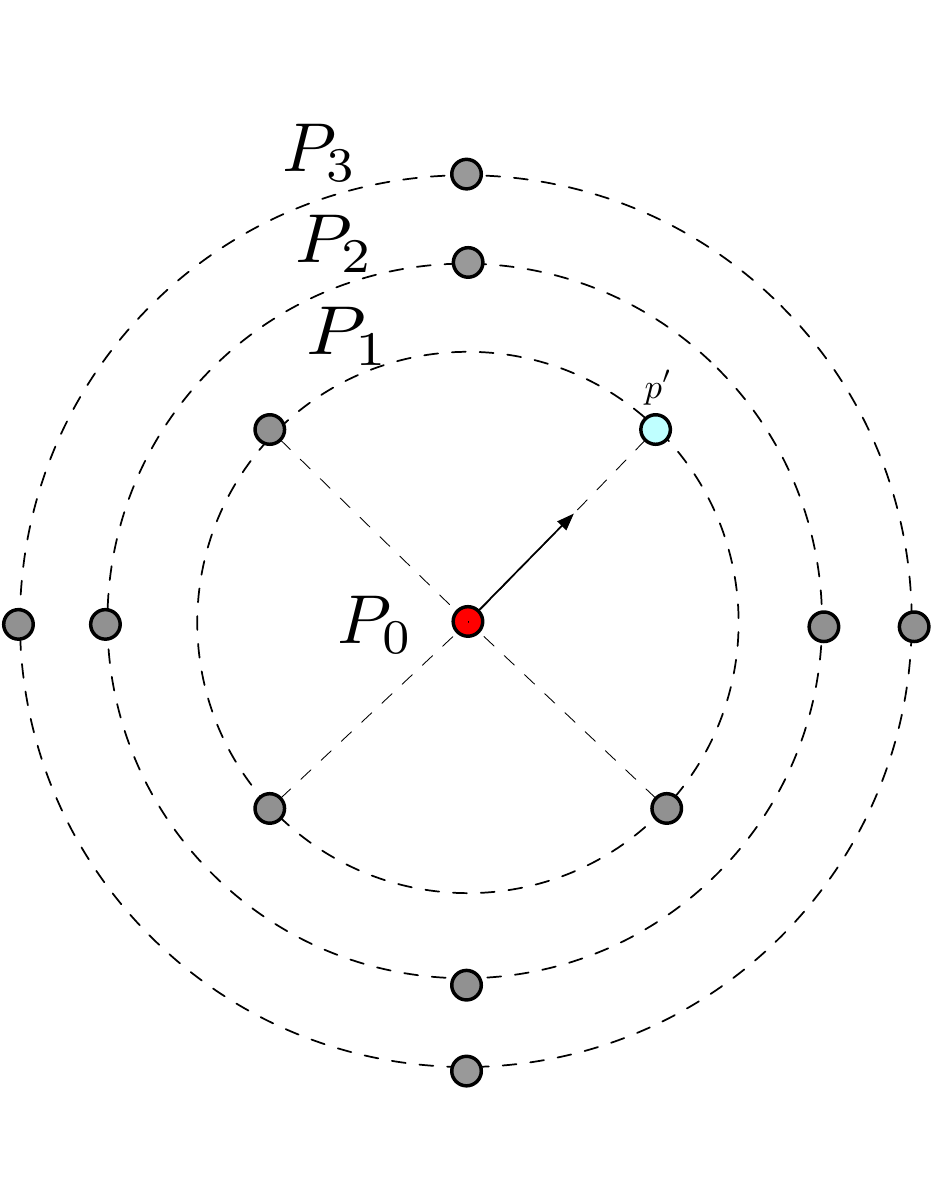}
  \caption{Movement of the central robot $r_c$: it moves towards position $p'$ of a distance that is much smaller than the radius of $P_1$. The pivot point $p'$ is on $P_1$}\label{fig:cmove}
\end{figure}

\item $n>3$  (Case C2): Robot $r_c$ chooses a robot position $p'$ on $P_1$ as pivot. It   moves on the segment connecting $c$ and $p'$  of a small distance (much smaller then the radius of $P_1$). See Figure \ref{fig:cmove}.
\end{itemize}

Robot $r_l$, when different from $r_c$, executes the following movements according to $n$ (such movements are computed by the function  {\sc ComputeMovementNotCentral}):
\begin{itemize}
\item $n=3$. (Case L1): (Robots  on a single line): Robot $r_l$ moves on $s$ by a small distance, 
creating a new segment $s'$. The pivot position $p'$ will be indicated by the direction that goes from the old center of $s$ to the new center of $s'$ (see Figure \ref{f:elements2}). 

\item $n>3$ (case L2): We have three sub-cases depending on which $P_j$ robot $r_l$ is positioned, and on the number of other robots on $P_j$. Note that $r_l \not\in P_0$ (otherwise it would be central). Remember that $P_m$ is the outermost circle, and it coincides with the {\sc SEC} of $C$. We treat each $P_j$ as a set, e.g. $|P_j|$ indicates the number of robots/positions in $P_j$. Let $x$ be the difference between the radii of $P_{j-1}$ and $P_{j}$, and let $h$ be the segment connecting $P_0$ and robot $r_l$. Let $p$ be the position of the first robot on $P_1$ encountered by walking in clockwise direction starting from the intersection between $h$ and $P_1$. Let $nhop$ be the number of positions between $p$ and the robot $p'$ that $r_l$ wants to indicate. Recall that $p'$ is the robot that $r_l$ will move towards if $r_l$ was in $P_0$. 

\begin{itemize}
\item  Sub-case (L2.1). When $P_j \neq P_m$ or $P_j=P_m$ and $|P_m| >3$:  
Robot $r_l$   moves on $h$ towards $P_j$ by a quantity $encode(nhop)*\frac{x}{2}$. Where $encode$ is an appropriate function from $\mathbb{N}$ to $(\frac{1}{2},1)$.  See Figure \ref{fig:notcentralmovel} for an example. 
\item Sub-case (L2.2). When $P_j=P_m$ and $|P_m| = 3$: Note that $P_m$ has to be rotationally symmetric; therefore it contains $3$ robots each of them forming an angle of $\frac{2\pi}{3}$ with its adjacent neighbours. Robot $r_l$ moves to a point of $P_j$ that creates with its counter-clockwise neighbour an angle that is $encode(nhop)*\frac{2\pi}{3}$, where $encode$ is an appropriate  function from $\mathbb{N}$ to $(\frac{1}{2},1)$. See Figure \ref{fig:notcentralmovel3r} for an example. 

\item Sub-case (L2.3).  When $P_j=P_m$ and $|P_m| =  2$: let $s$ be the segment connecting the two robots on $P_j$. Robot $r_l$ moves  on $s$, expanding $P_j$ in such a way  that the  new diameter is $2D+encode(nhop)*D$$. W$, where $encode$ is an appropriate  function from $\mathbb{N}$ to $(\frac{1}{2},1)$. See Figure \ref{fig:notcentralmovel2r} for an example. 
\end{itemize}
Note that is not possible to have $|P_m| =  1$, since $C$ is rotationally symmetric. 
\end{itemize}
It is easy to see that when $r_c$ (or $r_c$ and $r_l$)  move,  the resulting configuration $C'$ is not in ${\cal C}_{\odot}$.

\begin{figure}[H]
\begin{center}
\begin{subfigure}{0.45\textwidth}
  \includegraphics[width=0.75\textwidth]{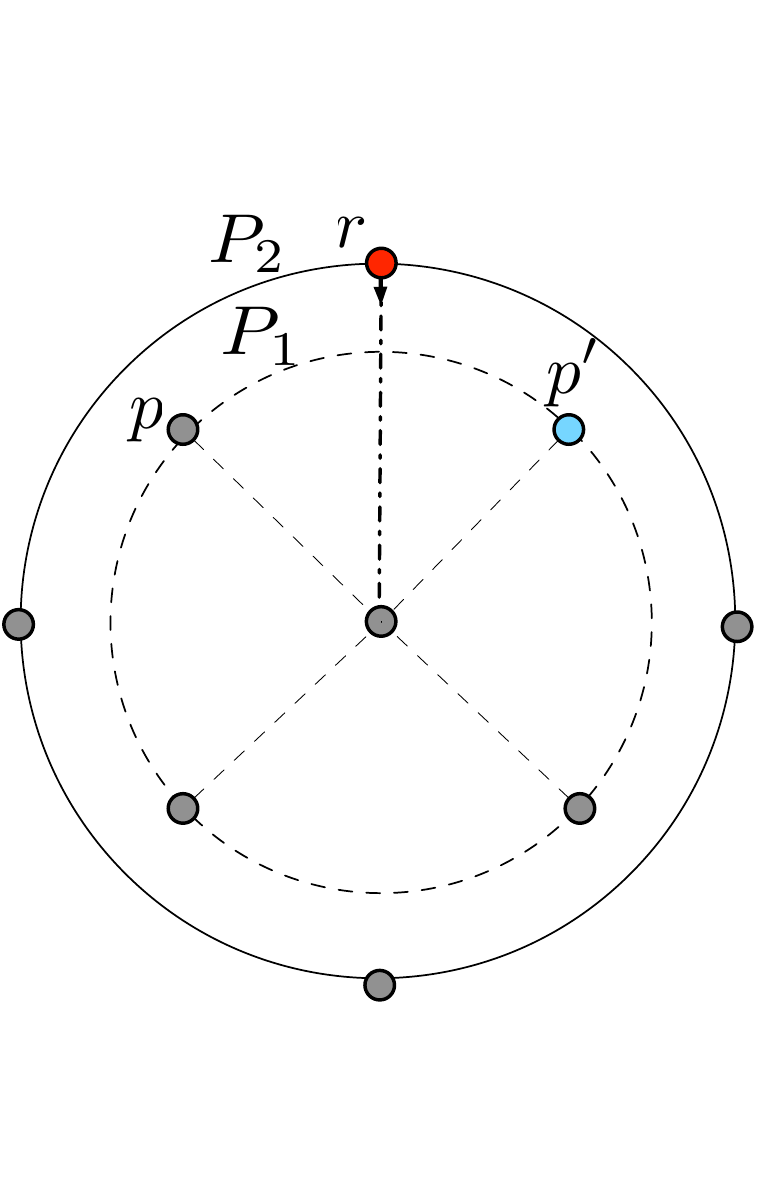}
  \vspace{-1cm}
  \caption{L2.1. The non-central leader robot is not on the {\sc SEC}, or it is on the {\sc SEC} with at least other three robots.}\label{fig:notcentralmovel}

\end{subfigure}
$\,\,$
\begin{subfigure}{0.45\textwidth}
  
  \includegraphics[width=1.3\textwidth]{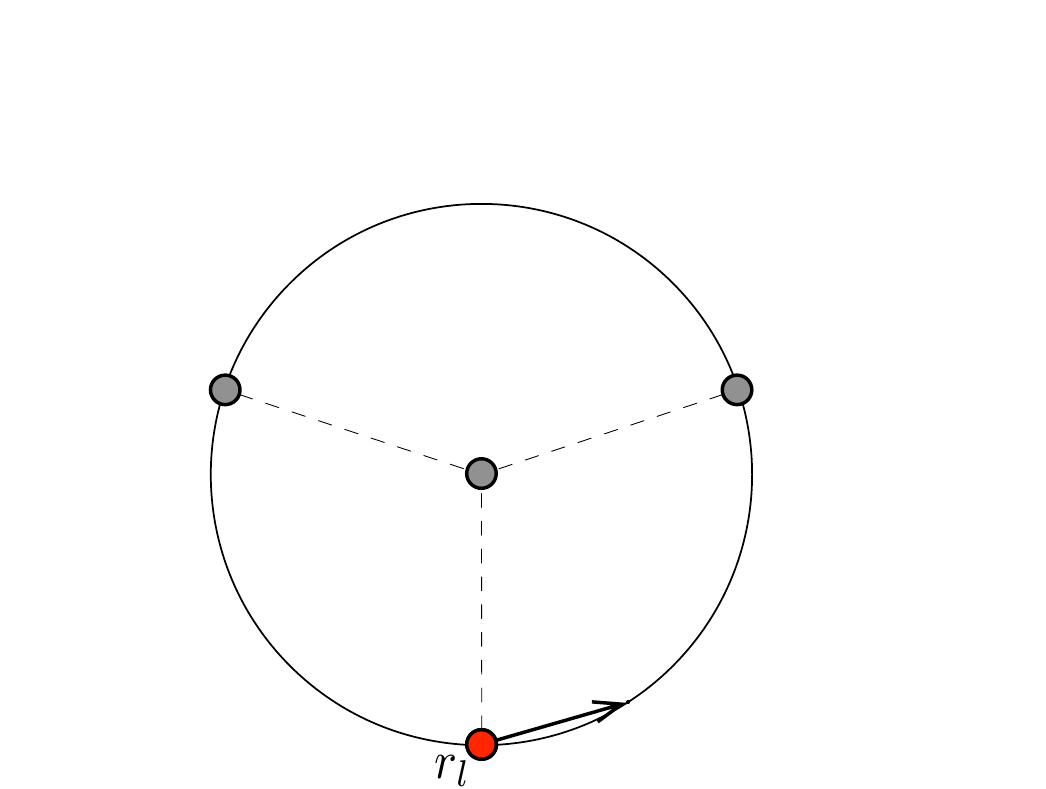}

  \caption{L2.2. The non-central leader robot is on the {\sc SEC}, and the {\sc SEC} contains exactly three robots.}\label{fig:notcentralmovel3r}

    \end{subfigure}%
\begin{center}
\begin{subfigure}{0.45\textwidth}
  \includegraphics[width=0.9\textwidth]{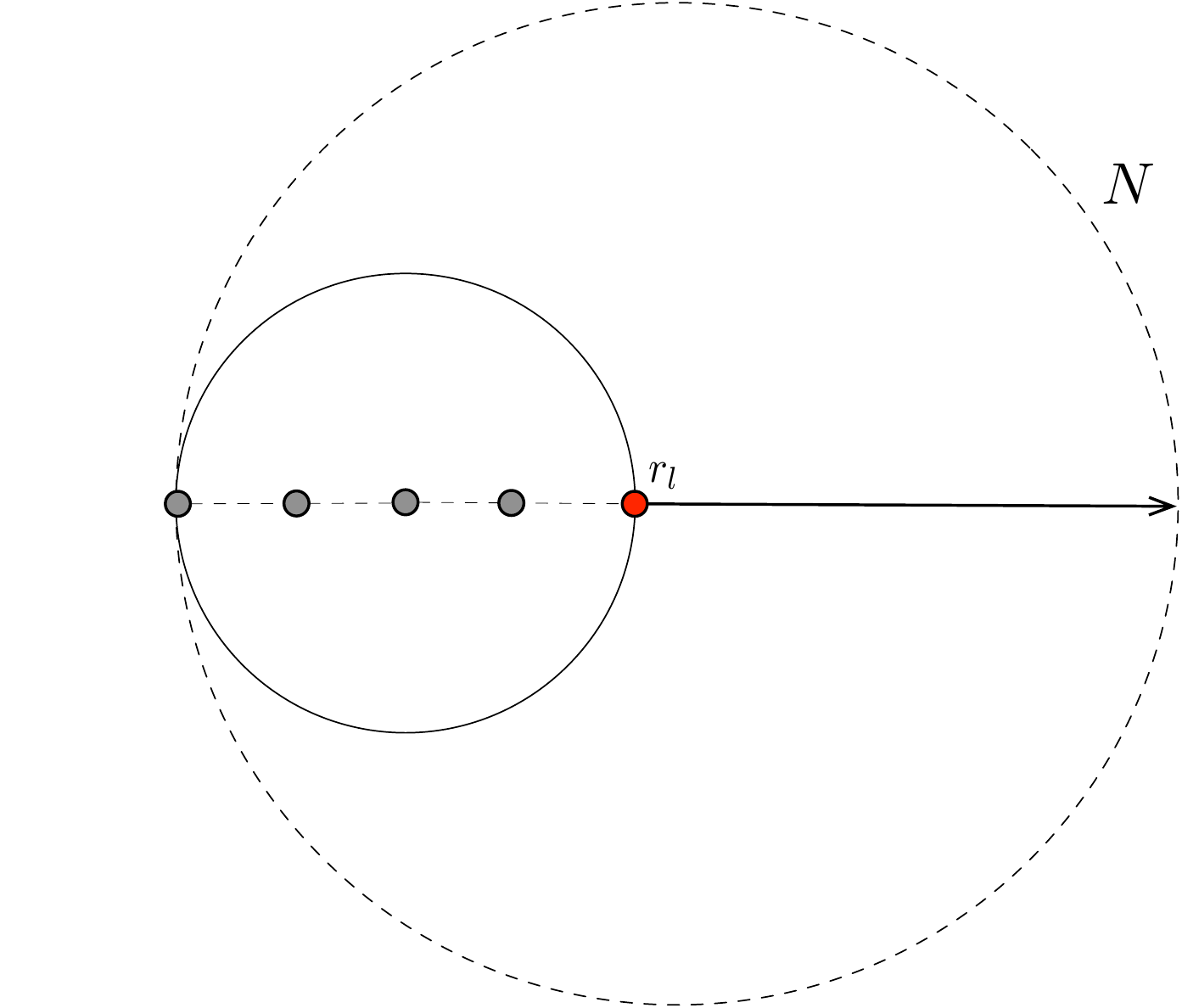}

  \caption{L2.3. The non-central leader robot is on the {\sc SEC}, and the {\sc SEC} contains exactly two robots.}\label{fig:notcentralmovel2r}
  \end{subfigure}%
\end{center}
\caption{Movements of non-central leader robot $r_l$ when $n \geq 4$. \label{figure:cmove}}
\end{center}
\end{figure}

\bigskip
\bigskip
\noindent {\bf Reconstruction of the initial configuration.}
When the current configuration $C'$ is not in  ${\cal C}_{\odot}$ and the robots have bit $b= 1$, the robots know that they are in an intermediate configuration and 
they have to (1) reconstruct the original configuration $C$, (2) determine the pivot point $p'$ in $C$, 
and (3) identify the leader robot $r_l$. The reconstruction is performed by procedure {\sc Reconstruct}, which again, depends on the value of $n$.

\begin{itemize} 

\item In case $n=3$: the robots are forming a triangle. The base of the triangle is its largest edge $e$. The algorithm uses the height of the triangle, w.r.t. the base $e$ to understand if $r_l$ was $r_c$ or not (see Figure \ref{f:elements3}). 

\begin{itemize}

\item If the height of the triangle is exactly half of $e$, the algorithm infers that $r_l=r_c$ and that the two other robots are the endpoints of $e$, case (C1). The pivot $p'$ is the point of the first clockwise robots encountered following the arc that
includes robot $r_l$ and the other two. The original configuration $C$ is easily reconstructed: the endpoints of $e$ are in the same position, and the central robot will be in the intersection of the perpendicular segment that goes through
$r_l$ and $e$. 

\item If the height is slightly less, or slightly more, than the largest edge $e$. Then the algorithm infers that $r_l$ was one of the endpoint; we are in case (L1.1). The reconstruction of $C$ is simple:  take the intersection $x$ of the perpendicular segment that goes through
$r_l$ and $e$, the position of the endpoints in $C$ is reconstructed using the fact that the hight of the triangle is exactly half of the original segment, and that $x$ was the center of the original segment. Robot $r_l$ is the endpoint that moved, and the 
position $p'$ is decided evaluating if $r_l$ moved towards or away from the old center.  
\end{itemize}
 
 \item In case $n>3$: The algorithm starts by examining $P_m$, in order to understand if $r_l$ was on the {\sc SEC} and executed the sub-case (L2.2) or (L2.3). If the test is negative it proceeds using an {\emph ``onion peeling''} approach, in which the algorithm, starting from the outermost $P_m$, progressively examines each $P_j$ until it finds an asymmetry or it reaches $P_0$. The onion peeling proceeds by first computing the SEC, that is $P_m$, and then   computing each   $P_j$ by finding progressively smaller concentric circles. 
 	\begin{itemize}
\item Test for case (L2.3): This test case is done only on $P_m$. If  the center of $P_m$ is not contained in {\sc SEC}$(C' \setminus P_{m-1})$, then the algorithm detects case (L2.3).   $P_m$ is adjusted to a new one that has the diameter equal to the distance between the two furthest robots on $P_m$. Robot $r_l$ will be the robot on $P_m$ that is farthest from robots in $P_{m-1}$. The reconstruction of the last layer is done by knowing that it will be a circle with the same center of {\sc SEC}$(C' \setminus P_{m-1})$ that passes through $P_m \setminus \{r_l\}$ and finally $p'$ will be indicated by decoding the information encoded in the diameter of $P_m$. 

\item Test for case (L2.2): This test case is done only on $P_m$. If $|P_m|=3$ and it is not rotationally symmetric, and ($|P_{m-1} |> 1$ or $m-1=0$), then the algorithm detects case  (L2.2). Robot $r_l$ is one that is not forming an angle of $\frac{2\pi}{3}$ radians with any of it adjacent neighbours, position $p'$ is encoded in the smallest angle that $r_l$ is forming. The original position of $r_l$ can be easily reconstructed: it is the one that forms an angle of $\frac{2\pi}{3}$ with each of its adjacent robots. 

\item Test for case (L2.1): This test case is done on layers different than $P_m$. If $|P_j|=1$ then the algorithm detects case (L2.1). Robot $r_l$ is the only robot in $P_j$ the reconstruction and the detection of $p'$ are trivial. 
\end{itemize}
If the algorithm reaches $P_0$ not finding an asymmetry then we have that $r_l=r_c$, case (C2). The decoding is trivial: the original position of $r_l$ is the center of $P_1$, position $p'$ is the one it moved towards. 
Note that the algorithm uses always robot $r_l$ to get robot $p'$, in case $r_l \neq r_c$ also $r_c$ moves. However, the original position of $r_c$ is always trivial to reconstruct (the center of any $P_j$ with $j >0$). 
 \end{itemize}

\begin{figure}[tbh]
    \begin{center}
\begin{subfigure}{0.4\textwidth}
  \includegraphics[width=\textwidth]{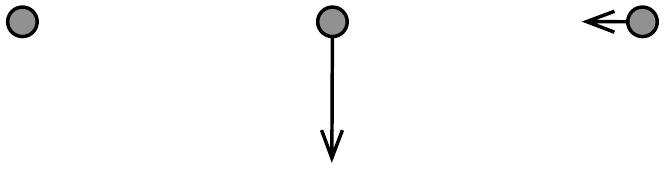}
  \caption{The case of $3$ robots, the central robot moves perpendicularly to initial line, and it will indicate the left robot as direction for the permutation.}
  \label{f:elements1}
\end{subfigure}
$\,\quad$
\begin{subfigure}{0.4\textwidth}
  
  \includegraphics[width=\textwidth]{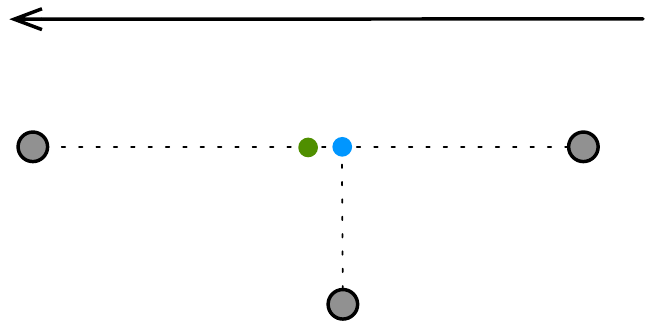}

  \caption{The external robot moved and it fixed as direction for permutation the one that goes  from the old center, blue dot, to the new one, green dot.     \label{f:elements2}}

    \end{subfigure}%

\begin{subfigure}{0.4\textwidth}
\includegraphics[width=\textwidth]{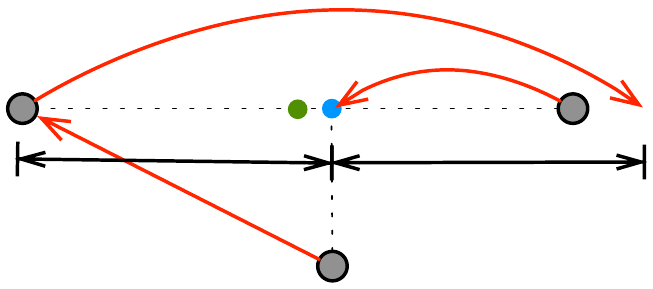}

  \caption{The exact positions of robots are easily reconstructable: the new central robot has to be places in the blue dot, and the left central robot has to move to its simmetric position with respect to the old center, the blue dot.     \label{f:elements3}}
  \end{subfigure}%
\end{center}
\caption{Special case $n=3$ robots.}
\label{f:special3}
\end{figure}

\begin{theorem}\label{tmmovell1bit}
There exists an universal algorithm to solve \visitall for  robots with $1$ bit of persistent memory.  
\end{theorem} 
\begin{proof}
We analyse the correctness of the proposed solution considering three cases:
\begin{itemize} 
\item {\bf When the robots start in a configuration $C \not\in{\cal C}_{\odot}$}: Note, that in this case all robots have bit $b$ set to $0$. This means, that the condition dictated by the {\tt If} condition  at line \ref{alg1bit:notcentral} is satisfied, thus the robots will fully execute Algorithm \ref{alg:visitallrho2}. Therefore, \visitall will be correctly solved. 

\item {\bf When the robots start in a configuration $C \in{\cal C}_{\odot}$ and $n>3$}. In this case,  for any central configuration reached we may alternate between two different scenario: (S1) only the central robot moves; (S2) two robots move. Scenario (S1) happens in the first configuration, and each time robot $r_l$ goes back to the central position. Scenario (S2) happens each time we are in a central configuration, but $r_l$ is not the central robot. We now prove the correctness for each of the aforementioned possibilities, and this is done by showing that robots start from a configuration $C \in \Pi(C_0)$, they reach a configuration $C' \not\in{\cal C}_{\odot}$, and then again a configuration $C'' \in \Pi(C_0)$ that is a cyclic permutation of $C_0$. 
\begin{itemize}
\item {\bf Scenario (S1), only the central robot moves.} At the beginning, the robots have bit $b=0$ and the configuration is ${\cal C}_{\odot}$, this implies that only the central robot $r_l$ move, see line \ref{alg1bit:centralrobot} of Algorithm \ref{alg:visitall1bit}, all the robots set their bit $b$. 
The central robot $r_l$ moves of a small distance $d$ towards a robot $p'$.
At the next round, the robots will se a $C' \not\in{\cal C}_{\odot}$, but since  $b=1$ they know that they started from a central configuration, see line \ref{alg1bit:wascentral}. Notice that,  if $r_l$ moved of a small distance still being inside the circle formed by the innermost robots (see as reference Figure \ref{fig:runcentral}) then it is simple for all robots to reconstruct the proper $C$ by seeing $C'$ and executing the {\sc Reconstruct} procedure. Therefore, they can understand that $r_l$ is the only robot that moved and everyone appoints $r_l$ as {\em leader}, it also simple to identify the pivot robot $p'$ in the innermost circle of $P_1$.

By using this direction, the robots compute a common cyclic order on  the positions in $C''$  (see Figure \ref{fig:runcentral}) and everyone moves according to this order and its position in $C''$, reaching a permutation of $C'' \in \Pi(C)$, before moving, all robots but the appointed leader update their bit $b$ (line \ref{alg1bit:unsetbit}).

\item {\bf Scenario (S2), only two robots move.} Robots that change position are the central robot $r_c$ (line \ref{alg1bit:centralrobot} of Figure \ref{alg:visitall1bit}), and the only one that started the round with bit $b=1$ that is robot $r_l$ (line \ref{alg1bit:notcentralrobot} of Figure \ref{alg:visitall1bit}). Before ending the round, all the  robots set their bit $b$  (line \ref{alg1bit:setbit}).
The movement of   robot $r_l$ is different depending  on its position. The purpose of this move is to indicate the pivot robot $p'$, that $r_l$ would have chosen if it was the central robot. 
 After $r_l$ and $c$ moved, we are in a configuration $C'$ and all robot have bit $b$ set, so they can remember they were starting from a central configuration. The only item left to show is that they can all reconstruct the central configuration $I$ and they can all agree on robot $r_l$ being the leader. 
We now show that the robots do reconstruct $C$ starting from $C'$ and  thus agree on the same cyclic order. 
 The reconstruction is done by the procedure {\sc Reconstruct}. We have to show that for each possible case the reconstruction is correct.
 
 \begin{itemize} 
\item  Reconstruction from a configuration created with move (L2.1).  First of all the algorithm cannot mistake it for cases (L2.2) and (L2.3). The correct reconstruction is immediate.

\item  Reconstruction from a configuration created with move (L2.2). The case cannot be mistaken for (L2.3): the {\sc SEC} is concentric with the other $P_j$. It cannot be mistaken for case L2.1: $|P_{m-1}| >1$ or $m-1=0$. The reconstruction is immediate. 

\item  Reconstruction from a configuration created with move (L2.3).  This case is detected by observing that the centre of {\sc SEC} is not contained in {\sc SEC}$(C' \setminus P_{m})$. This happens only when $r_l$ moved according to rule (L2.3).
The {\sc Reconstruct} adjusts the {\sc SEC}, that contains three robots, to a new sec $P_{m}$ that has as diameter the two furthest robots on the old {\sc SEC}. Let $r_l$ and $r_x$ be the two robots on the new $P_m$. We argue that the distance of this two robots is encoding the position $p'$.
Refer to Figure \ref{fig:notcentralmovel2r}: robot $r_l$ moves from the other robots up to a distance $2D+x*D$. By doing so,  it is becoming the robot that is furthest from all the others in $C'$. It is also clear that the by removing $r_l$ and $r_x$ we get $C \setminus \{r_l,r_x\}$ 
(neglecting for now the central robot). The diameter $D$ can be recovered by taking the circle with center in {\sc SEC}$(C \setminus \{r_l,r_x\})$ and passing through $r_x$. With this information the decoding is simple, and the reconstruction of $C$ is immediate.

 \end{itemize}
 
  \item  {\bf When the robots start in a configuration $C \in{\cal C}_{\odot}$ and $n=3$}: In this case when the configuration is in ${\cal C}_{\odot}$ the movements of the central robot and of the leader robot $r_l$ follow special rules. Notice that all three robots lie on a segment $s$.  Robot $r_l$ moves following case (L2.1) if is not central. The central robot moves according to  case (C2).  It   easy to see that the  robots reconstruct the central configuration and they correctly elect the leader, which is either the central robot (if it is the only robot that moved)  or the external robot (if two robots moved).
It is also immediate that \visitall is solved, since robots execute continuously the same cyclic permutation. 
 \end{itemize}
Once $C,p',r_l$ are known the robots agree on the same cyclic order and reach a new configuration $C \in \Pi(C_0)$. 
An immediate induction shows that starting from the initial configuration $C_0$ we alternate either in scenario S1 or S2 solving \visitall.\end{itemize}\end{proof}

\section{Concluding Remarks}

To the best of our knowledge, this is the first investigation of the problems of permuting the positions of a set of mobile robots in the plane. Surprisingly this class of problems seems to be more difficult than the previously studied problems such as gathering and pattern formation, which have easy solutions for the strongest model of fully synchronous robots with rigid movements. Thus the characterization of solvable instances for permutation problems is quite different as shown in this paper. Moreover we also showed that being non-oblivious is helpful for permuting robots, unlike the formation problems where the solvability is unaffected by obliviousness~\cite{yamashita2010}.
The paper opens several research directions that are worth investigating: an interesting direction would be to discover other class of problems which cannot be solved even when it is easy to elect a leader (as the class of problems considered here). The difficulty in solving the permutation problems seems to be unrelated to agreement problems such as leader election. In particular we may try to study the differences between leader election and permutation problems and determine if the latter is strictly more difficult than the former. We may also consider other interesting assumptions that can help in overcoming the challenges for permuting robots without orientation. One possibility is the investigation of robots with the additional capability of communicating using visible lights \cite{DaFPSY16,inf,DiVi19}.\\

\noindent {\bf Acknowledgements:} This work has been partially funded by the University of Rome ``La Sapienza'' with the Calypso project. Some of the authors were supported in part by NSERC Canada through the Discovery Grants program and by Prof. Flocchini's University Research Chair.

\end{document}